\documentclass[a4paper,onecolumn,10pt]{amsart}

\usepackage{amsmath}
\usepackage{amssymb}  
\usepackage{amsthm}
\usepackage{amsfonts}
\usepackage{graphicx}
\usepackage{cancel}
\usepackage{tikz}
\usetikzlibrary{cd, arrows, shapes}
\usepackage{verbatim}

\usepackage{bbm}
\usepackage{url}
\usepackage{enumerate} 
\usepackage{ upgreek }
\usepackage{dsfont}
\usepackage{color}
\usepackage{subcaption}

\usepackage{makecell}
\usepackage{diagbox}
\usepackage{wrapfig}
\usepackage{rotating}
\usepackage{tabularx}

\usepackage{hyperref}
\hypersetup{
    colorlinks=true,
    linkcolor=blue,
    citecolor=red,
    filecolor=magenta,      
    urlcolor=cyan,
}

\newcommand{\ket}[1]{| #1 \rangle}

\newcommand{\bra}[1]{\langle #1 |}

\newcommand{\braket}[2]{\langle #1 , #2 \rangle}

\newcommand{\proj}[2]{| #1 \rangle\!\langle #2 |}

\newcommand{\id}{\ensuremath{\mathds{1}}}

\newcommand{\cC}{\mathcal{C}}

\newcommand{\cH}{\mathcal{H}}
\newcommand{\cI}{\mathcal{I}}

\newcommand{\cK}{\mathcal{K}}
\newcommand{\cL}{\mathcal{L}}
\newcommand{\cM}{\mathcal{M}}

\newcommand{\cS}{\mathcal{S}}

\newcommand{\cV}{\mathcal{V}}
\newcommand{\cW}{\mathcal{W}}

\def\beq{\begin{equation}}
\def\eeq{\end{equation}}
\def\bq{\begin{quote}}
\def\eq{\end{quote}}
\def\ben{\begin{enumerate}}
\def\een{\end{enumerate}}
\def\bit{\begin{itemize}}
\def\eit{\end{itemize}}

\def\ra{\rightarrow}

\def\lb{\left(}
\def\rb{\right)}
\def\lset{\lbrace}
\def\rset{\rbrace}

\def\r|{\right|}
\def\lbr{\left[}
\def\rbr{\right]}
\def\ident{\textnormal{id}}
\def\one{\id}

\newcommand\C{\mathbbm{C}}

\newcommand\R{\mathbbm{R}}
\newcommand\N{\mathbbm{N}}

\newcommand{\Tr}{\text{Tr}}

\DeclareMathOperator{\Pos}{Pos}

\DeclareMathOperator{\EB}{EB}
\DeclareMathOperator{\LorFact}{LorFact}
\DeclareMathOperator{\LorEA}{LorEA}
\DeclareMathOperator{\maxEA}{maxEA}
\DeclareMathOperator{\LorEB}{LorEB}
\DeclareMathOperator{\hs}{hs}
\DeclareMathOperator{\Aut}{Aut}
\DeclareMathOperator{\Nuc}{Nuc}

\DeclareMathOperator{\conv}{conv}
\theoremstyle{plain}
\newtheorem{thm}{Theorem}[section]
\newtheorem{lem}[thm]{Lemma}
\newtheorem{cor}[thm]{Corollary}
\newtheorem{prop}[thm]{Proposition}
\newtheorem{defn}[thm]{Definition}

\theoremstyle{definition}
\newtheorem{example}{Example}

\begin{document}
\title{{Annihilating and breaking Lorentz cone entanglement}}

\author{Francesca La Piana}
\address{Department of Mathematics, University of Oslo}
\email{franla@math.uio.no}

\author{Alexander M\"uller-Hermes}
\address{Department of Mathematics, University of Oslo}
\email{muellerh@math.uio.no}

\date{\today}

\begin{abstract}
Linear maps between finite-dimensional ordered vector spaces with orders induced by proper cones $\cC_A$ and $\cC_B$ are called entanglement breaking if their partial application sends the maximal tensor product $\cK\otimes_{\max} \cC_A$ into the minimal tensor product $\cK\otimes_{\min} \cC_B$ for any proper cone $\cK$. We study the larger class of Lorentz-entanglement breaking maps where $\cK$ is restricted to be a Lorentz cone of any dimension, i.e., any cone over a Euclidean ball. This class of maps appeared recently in the study of asymptotic entanglement annihilation and it is dual to the linear maps factoring through Lorentz cones. Our main results establish connections between these classes of maps and operator ideals studied in the theory of Banach spaces. For operators $u:X\ra Y$ between finite-dimensional normed spaces $X$ and $Y$ we consider so-called central maps which are positive with respect to the cones $\cC_A=\cC_X$ and $\cC_B=\cC_Y$. We show how to characterize when such a map factors through a Lorentz cone and when it is Lorentz-entanglement breaking by using the Hilbert-space factorization norm $\gamma_2$ and its dual $\gamma^*_2$. We also study the class of Lorentz-entanglement annihilating maps whose local application sends the Lorentzian tensor product $\cC_A\otimes_{L} \cC_A$ into the minimal tensor product $\cC_B\otimes_{\min} \cC_B$. When $\cC_A$ is a cone over a finite-dimensional normed space and $\cC_B$ is a Lorentz cone itself, the central maps of this kind can be characterized by the $2$-summing norm $\pi_2$. Finally, we prove interesting connections between these classes of maps for general cones, and we identify examples with particular properties, e.g., cones with an analogue of the $2$-summing property.

\end{abstract}

\maketitle

\section{Introduction}

Convex cones naturally arise in quantum physics, when studying properties of entanglement~\cite{horodecki2009quantum}. Here, the most important examples are the cones of positive semidefinite $d\times d$ matrices $M_d(\C)^+$ inside the real vector space $M_d(\C)_{sa}$ of selfadjoint $d\times d$ matrices with entries in $\C$. Recently, the study of entanglement has been extended to other settings involving more general examples of convex cones (see, e.g.,~\cite{lami2018non,plavala2023general} and references therein). 

Consider proper cones $\cC_A\subseteq \cV_A$ and $\cC_B\subseteq \cV_B$ in finite-dimensional real vector spaces $\cV_A$ and $\cV_B$, respectively. A convex cone $\cC\subseteq \cV$ is called \emph{proper} if it is closed, pointed (i.e., $\cC\cap (-\cC)=\lset 0\rset$), and generating (i.e., $\cC+(-\cC)=\cV$). The cones of positive semidefinite matrices mentioned above are examples of proper cones. To study entanglement we introduce cones in the algebraic tensor product $\cV_A\otimes \cV_B$ derived from the cones $\cC_A$ and $\cC_B$. The \emph{minimal tensor product} is given by
\[
\cC_A\otimes_{\min} \cC_B = \conv\lset x\otimes y ~:~x\in \cC_A , y\in \cC_B\rset ,
\]
and the \emph{maximal tensor product} is defined as 
\[
\cC_A\otimes_{\max} \cC_B = (\cC^*_A\otimes_{\min} \cC^*_B)^* .
\]
Here, $\cC^*\subseteq \cV^*$ denotes the dual cone of $\cC\subseteq \cV$, i.e., the cone of functionals $f:\cV\ra \R$ satisfying $f(x)\geq 0$ for any $x\in\cC$. The terminology "minimal" and "maximal" used here, refers to the size of the respective set, and we have
\[
\cC_A\otimes_{\min} \cC_B\subseteq \cC_A\otimes_{\max} \cC_B ,
\]
for any pair of proper cones. Analogous to separable states in quantum physics, we think of $\cC_A\otimes_{\min} \cC_B$ as the set of separable tensors. All other tensors in $\cC_A\otimes_{\max} \cC_B$ are considered to be entangled. Entanglement exists for a given pair of cones if and only if neither of them has a simplicial base~\cite{aubrun2021entangleability}.

\subsection{Entanglement annihilating and entanglement breaking maps} 

Interesting questions arise from studying how positive maps can transform different types of entanglement. A linear map $P:\cV_A\ra \cV_B$ is called \emph{$(\cC_A,\cC_B)$-positive} if $P(\cC_A)\subseteq \cC_B$, and we will simply call it \emph{positive} if it is clear from context which cones are used. Motivated by the class of entanglement breaking maps in the context of quantum entanglement~\cite{horodecki2003entanglement}, a linear map $P:\cV_A\ra \cV_B$ is called \emph{$(\cC_A,\cC_B)$-entanglement breaking} if 
\begin{equation}\label{equ:EBDef}
(\ident_{\cW} \otimes P)(\cK\otimes_{\max}\cC_A)\subseteq \cK\otimes_{\min}\cC_B ,
\end{equation}
for any proper cone $\cK\subseteq \cW$. Again, we will simply call a map entanglement breaking if it is clear from context which cones are used. A linear map $P:\cV_A\ra \cV_B$ is $(\cC_A,\cC_B)$-entanglement breaking if and only if it can be written as a convex combination of rank-$1$ maps $x\cdot f$ with $x\in \cC_B$ and $f\in \cC^*_A$ (see, e.g., \cite[Proposition 2.2.]{aubrun2023annihilating}). We denote by $\Pos(\cC_A,\cC_B)$ and $\EB(\cC_A,\cC_B)$ the convex cones of positive linear maps and entanglement breaking maps, respectively. 

A $(\cC_A,\cC_B)$-positive map might destroy entanglement by acting locally on each factor of the tensor product $\cC_A\otimes_{\max} \cC_A$. Specifically, we call a linear map $P\in \Pos(\cC_A,\cC_B)$ \emph{$2$-max-entanglement annihilating} if 
\[
(P\otimes P)\lb \cC_A\otimes_{\max} \cC_A\rb\subseteq \cC_B\otimes_{\min} \cC_B .
\] 
Again, this definition is motivated by a similar class of maps studied in quantum information theory (see~\cite{moravvcikova2010entanglement}). The set $\maxEA_2(\cC_A,\cC_B)$ of all $2$-max-entanglement annihilating maps contains the convex cone $\EB(\cC_A,\cC_B)$ of entanglement breaking maps. However, in general $\maxEA_2(\cC_A,\cC_B)$ is not convex (see Appendix~\ref{app:maxEANotConv}), and $\EB(\cC_A,\cC_B)\subsetneq \maxEA_2(\cC_A,\cC_B)$. More generally, we can consider \emph{$k$-max-entanglement annihilating} maps $P\in \Pos(\cC_A,\cC_B)$ for any $k\in\N$ defined by 
\[
P^{\otimes k}\lb \cC^{\otimes_{\max} k}_A\rb\subseteq \cC^{\otimes_{\min} k}_B .
\]
In this article we will not consider $k$-max-entanglement annihilating maps for $k>2$, and in the following we will often say max-entanglement annihilating to denote the $k=2$ case.

\subsection{Breaking and annihilating Lorentz cone entanglement}

Entanglement breaking maps are characterized by \eqref{equ:EBDef}, i.e., the property that they break entanglement between the input cone $\cC_A$ and any reference cone $\cK$. Natural classes of positive maps arise by restricting the reference cones $\cK$ in \eqref{equ:EBDef} to cones from a specific family of proper cones. In this article, we will consider the family of Lorentz cones (also known as second-order cones): They are defined as
\[
\cL_n = \lset (t,x)\in \R\times \ell^n_2~:~t\geq \|x\|_2\rset ,
\]
for any $n\in\N$, where $\ell^n_2 = (\R^n,\|\cdot\|_2)$ is the $n$-dimensional Euclidean space. Lorentz cones are proper and have other nice properties (see Section \ref{sec:LorentzCones}). For proper cones $\cC_A\subseteq \cV_A$ and $\cC_B\subseteq \cV_B$, we call a linear map $P:\cV_A\ra \cV_B$ \emph{Lorentz-entanglement breaking} if 
\begin{equation}\label{equ:LorEB}
(\ident_{k+1}\otimes P)\lb \cL_k\otimes_{\max} \cC_A\rb\subseteq \cL_k\otimes_{\min} \cC_B ,
\end{equation}
holds for all $k\in\N$. It is clear that the Lorentz-entanglement breaking maps form a closed convex cone. We denote this cone by $\LorEB(\cC_A,\cC_B)$ and note that we have the inclusions
\[
\EB(\cC_A,\cC_B)\subseteq \LorEB(\cC_A,\cC_B)\subseteq \Pos(\cC_A,\cC_B)
\]
It is easy to show that the cone $\LorEB(\cC_A,\cC_B)$ is pointed and generating, and hence a proper cone in the vector space $\cL\lb \cV_A,\cV_B\rb$ of linear maps between $\cV_A$ and $\cV_B$. We also note that the Lorentz-entanglement breaking maps have an ideal property in the following sense: For proper cones $\cC_{A'}$ and $\cC_{B'}$ and positive maps $A\in \Pos(\cC_{A'},\cC_{A})$ and $B\in \Pos(\cC_{B},\cC_{B'})$ we have $BPA\in \LorEB(\cC_{A'},\cC_{B'})$ for any $P \in \LorEB(\cC_A,\cC_B)$. This shows that $\LorEB(\cC_A,\cC_B)$ is a mapping ideal, a notion we introduce below. 

To study the structure of the cone $\LorEB(\cC_A,\cC_B)\in \cL\lb \cV_A,\cV_B\rb$ it is helpful to look at its dual cone inside the space of linear maps $\cL\lb \cV_B,\cV_A\rb$. Throughout this article, we define duals of cones of linear maps using the trace functional: Given a cone $\cM\subseteq \Pos(\cC_A,\cC_B)$ of positive maps, we define its trace dual cone $\cM^*$ as the set of all linear maps $Q\in \cL(\cV_B,\cV_A)$ such that 
\[
\Tr\lbr PQ\rbr \geq 0 ,
\]
for all $P\in \cM$. Here, the trace is defined on $\cL(\cV_B,\cV_B)$. In Theorem~\ref{thm:traceDualLorEBAndClosure} we show the trace dual cone of $\LorEB(\cC_A,\cC_B)$ to be the cone of \emph{Lorentz factorizable maps} given by 
\begin{equation}\label{equ:LorFactDef}
\LorFact(\cC_B,\cC_A) := \text{conv}\lset AB ~:~ B\in \Pos(\cC_B,\cL_k), A\in \Pos(\cL_k,\cC_A) , k\in\N\rset .
\end{equation}
This cone is closed and proper, and the dimension $k$ can be restricted to the minimum of $\dim(\cV_A)$ and $\dim(\cV_B)$. For interesting examples of Lorentz factorizable maps between cones of positive semidefinite matrices see~\cite[Section 5.2.]{aubrun2023annihilating}.

It will be helpful to identify linear maps $\cL(\cV_A,\cV_B)$ with tensors in the algebraic tensor product $\cV^*_A\otimes \cV_B$. Proper cones of linear maps give rise to tensor products between cones via this correspondence. In the case of the Lorentz factorizable maps, this gives rise to the following tensor product:

\begin{defn}[Lorentzian tensor product]
For proper cones $\cC_A$ and $\cC_B$ we define the \emph{Lorentzian tensor product} as 
\[
\cC_A\otimes_{L} \cC_B := \conv\lset (A\otimes B)\lb \hat{I}_k\rb ~:~ A\in \Pos(\cL_k,\cC_A),B\in \Pos(\cL_k,\cC_B), k\in\N \rset ,
\]
where $\hat{I}_k=\sum^k_{i=0} e_i\otimes e_i$ denotes the identity tensor in $\cL_k\otimes_{\max} \cL_k$.
\end{defn}

The Lorentzian tensor product defines a tensor product in the category of ordered vector spaces since 
\[
\cC_A\otimes_{\min} \cC_B\subseteq \cC_A\otimes_{L} \cC_B \subseteq \cC_A\otimes_{\max} \cC_B .
\]
It associates to any pair of proper cones $\cC_A\subseteq \cV_A$ and $\cC_B\subseteq \cV_B$ a reasonable cross cone (see~\cite{de2020tensor}) and it is an example of a tensor cone (see~\cite[Section 8]{de2020tensor}), i.e., the ordered vector space analogue of a tensor norm in the category of Banach spaces. Like similar constructions for Banach spaces, the Lorentzian tensor product is not associative (see Appendix~\ref{sec:NonAssoc}).   

Similar to the $2$-max-entanglement annihilating maps introduced above, we can consider linear maps annihilating the entanglement described by the Lorentzian tensor product. A linear map $P\in \Pos(\cC_A,\cC_B)$ is called \emph{Lorenz-entanglement annihilating} if 
\[
(P\otimes P)\lb \cC_A\otimes_{L} \cC_A\rb\subseteq \cC_B\otimes_{\min} \cC_B .
\]
The set of these maps will be denoted by $\LorEA_2(\cC_A,\cC_B)$, and we do not know whether it is convex in general. However, in the special case where $\cC_B=\cL_n$ is a Lorentz cone we show that the set $\LorEA_2(\cC_A,\cL_n)$ is a convex cone.

\subsection{Main results}

For a finite-dimensional normed space $X$ over $\R$ we define the convex cone 
\[
\cC_X = \lset (t,x)\in \R\times X ~:~t\geq \|x\|\rset .
\]
The cones $\cC_X$ are proper and for $X=\ell^n_2$ we recover the Lorentz cones introduced above. An important subset of positive maps with respect to cones $\cC_X$ and $\cC_Y$ over finite-dimensional normed spaces $X$ and $Y$ is given by the so-called \emph{central maps} preserving the direct sum structure in the definition of $\cC_X$. Specifically, a linear map $P:\R\times X\ra \R\times Y$ is called central if there are $\lambda\in\R$ and $u:X\ra Y$ such that $P((t,x))=(\lambda t,u(x))$. We write $P=\lambda\oplus u$ in this case. It was observed in~\cite[Proposition 2.25.]{lami2018non} that properties of a central map $P=\lambda\oplus u$ can be studied using norms of the map $u$: Specifically, the map $P$ is positive if and only if the operator norm satisfies $\|u\|\leq \lambda$ and $P$ is entanglement breaking if and only if $\text{Nuc}(u)\leq \lambda$. Here, $\text{Nuc}(u)$ denotes the nuclear norm of the operator $u$. It is a natural question, whether other properties of positive maps can be characterized in a similar way. 

For a linear map $u:X\ra Y$ the $2$-summing norm $\pi_2(u)$ is defined as the smallest $C\geq 0$ satisfying 
\[
\lb\sum_{i} \|u(x_i)\|^2\rb^{1/2} \leq C \max_{f\in B_{X^*}} \lb \sum_i |f(x_i)|^2\rb^{1/2} ,
\] 
for every finite subset $\lset x_i\rset$ of $X$. When $X$ and $Y$ are Euclidean spaces, the $2$-summing norm $\pi_2(u:X\ra Y)$ equals the Hilbert-Schmidt norm $\text{hs}(u)$, see~\cite[Proposition 10.1.]{tomczak1989banach}. Our first main result relates the $2$-summing norm of maps into Euclidean spaces to the annihilation of Lorentz-cone entanglement. 

\begin{thm}\label{thm:main2}
Let $u:X\ra \ell^n_2$ be a linear map and $\lambda\in \R$. We have $\lambda\oplus u\in\LorEA_2(\cC_X,\cL_n)$ if and only if $\pi_2(u) \leq \lambda$.
\end{thm}

Theorem \ref{thm:main2} characterizes the central Lorentz-entanglement annihilating maps when the output cone is a Lorentz cone in terms of the $2$-summing norm. Our second main result relates the $2$-summing norm of a linear map $u:X\ra Y$ for general finite-dimensional normed spaces $X$ and $Y$ to a particular factorization property of the associated central map. 

\begin{thm}
For $u:X\ra Y$ and $\lambda\in \R$, the following are equivalent:
\begin{enumerate}
\item We have $\pi_2(u) \leq \lambda$.
\item For any positive map $S\in \Pos(\cL_m,\cC_X)$ there exists a positive map $P\in \Pos(\cL_k,\cC_Y)$ and $Q\in \maxEA_2(\cL_m,\cL_k)$ such that $(\lambda\oplus u)S = PQ$.
\end{enumerate}
\end{thm}

Finally, we will consider the classes of Lorentz factorizable maps and maps that break entanglement with Lorentz cones. For a linear map $u:X\ra Y$ the \emph{Hilbert-space factorization norm}~\cite[Theorem 13.9.]{tomczak1989banach} is given by
\[
\gamma_2(u) = \inf \lset \|v_1\|\|v_2\|~:~u=v_2v_1, v_1:X\ra \ell^k_2, v_2:\ell^k_2\ra Y, k=\text{rk}(u)\rset .
\] 
Its trace-dual is the so-called \emph{2-dominated norm} $\gamma^*_2$. Remarkably, the 2-dominated norm can be characterized itself via a Hilbert-space factorization (see~\cite[Theorem 13.9.]{tomczak1989banach}) and we have 
\begin{equation}\label{equ:gamma2Star}
\gamma^*_2(u) = \inf \lset \pi_2(v_1)\pi_2(v^*_2)~:~u=v_2v_1, v_1:X\ra \ell^k_2, v_2:\ell^k_2\ra Y\rset .
\end{equation}
Following, the strategy from above, we can ask whether these norms characterize properties of central maps. We show the following characterization:

\begin{thm}\label{thm:Main1}
For any $u:X\ra Y$ and any $\lambda\in \R$ we have the following:
\begin{enumerate}
\item We have $\lambda\oplus u\in \LorFact(\cC_X,\cC_Y)$ if and only if $\gamma_2(u)\leq \lambda$.
\item We have $\lambda\oplus u\in \LorEB(\cC_X,\cC_Y)$ if and only if $\gamma^*_2(u)\leq \lambda$.
\end{enumerate}
\end{thm}

Since the property of a central map being Lorentz-entanglement breaking is characterized by the $2$-dominated norm it is natural to ask whether a factorization similar to \eqref{equ:gamma2Star} holds for all Lorentz-entanglement breaking maps. We leave this question for future research, but we note the following partial result:  

\begin{thm}
For proper cones $\cC_A\subseteq \cV_A$ and $\cC_B\subseteq \cV_B$ let $A\in \LorEA_2\lb \cC_A,L_m\rb$ and $B\in \LorEA_2\lb \cC^*_B,L_m\rb$ for some $m\in\N$. Then, we have $B^*A\in\LorEB\lb \cC_A,\cC_B\rb$. 
\end{thm} 

It is a natural problem to identify the pairs of proper cones $\cC_A$ and $\cC_B$ for which the inclusions 
\begin{equation}\label{equ:InclPosLorFact}
\Pos\lb \cC_A,\cC_B\rb \subseteq \LorFact\lb \cC_A,\cC_B\rb,
\end{equation}
or
\begin{equation}\label{equ:InclLorEBEB}
\LorEB\lb \cC_A,\cC_B\rb \subseteq \EB\lb \cC_A,\cC_B\rb,
\end{equation}
are strict. Combining Theorem \ref{thm:Main1} with results from~\cite{aubrun2021asymptotic} (together with the main result from \cite{aubrun2024limit}) and \cite{aubrun2024limit}, we can identify many examples of finite-dimensional normed spaces $X$ and $Y$ for which this is the case when $\cC_A=\cC_X$ and $\cC_B=\cC_Y$. Specifically, we see that both inclusions are strict when $X=Y$ is not Euclidean, and from a conjecture in~\cite{aubrun2024limit} it would follow that both inclusions are strict if and only if neither $X$ nor $Y$ is Euclidean. Going beyond cones over normed spaces, we consider cones of positive semidefinite matrices in Section~\ref{sec:PSDcones}. Our Theorem~\ref{thm:traceDualLorEBAndClosure} shows that
\[
\LorEB(M_3(\C)^+,M_3(\C)^+)^*=\LorFact(M_3(\C)^+,M_3(\C)^+) ,
\]
and in particular no closure is needed in~\eqref{equ:LorFactDef}. Combining this with the main result from~\cite{fawzi2019representing} shows that the inclusions \eqref{equ:InclPosLorFact} and \eqref{equ:InclLorEBEB} are strict in this case. However, it turns out to be more challenging to construct explicit examples of maps in $\LorEB(M_3(\C)^+,M_3(\C)^+)$ that are not entanglement breaking. In Theorem \ref{thm:PSDFactorizationEB} we show that such examples cannot be obtained by factoring through Lorentz-entanglement breaking central maps between cones $\cC_X$ and $\cC_Y$ for any choice of finite-dimensional normed spaces $X$ and $Y$. Using the counterexample to Peres conjecture presented in~\cite{vertesi2014disproving}, we construct an explicit example of a Lorentz-entanglement breaking map between cones of $3\times 3$ positive semidefinite matrices $M_3(\C)^+$ that is not entanglement breaking.

\subsection{Connection to asymptotic entanglement annihilation}

Theorem \ref{thm:Main1} fits nicely into a larger research program trying to characterize $\infty$-entanglement annihilating maps~\cite{moravvcikova2010entanglement,aubrun2023annihilating}. A positive map $P\in \Pos(\cC_A,\cC_B)$ is called \emph{$\infty$-max-entanglement annihilating} if
\[
P^{\otimes k}(\cC^{\otimes_{\max} k}_A)\subseteq \cC^{\otimes_{\min} k}_B ,
\]
for all $k\in\N$. Entanglement breaking maps are $\infty$-max-entanglement annihilating, but it is unknown whether there are $\infty$-max-entanglement annihilating maps that are not entanglement breaking for any pair of proper cones. Finding examples of $\infty$-max-entanglement annihilating maps that are not entanglement breaking between cones of positive semidefinite matrices has been linked to some long-standing conjectures in quantum information theory (see~\cite{aubrun2023annihilating}). It was shown in~\cite{aubrun2023annihilating} that any $\infty$-max-entanglement annihilating map $P\in \Pos(\cC_A,\cC_B)$ is necessarily Lorentz-entanglement breaking. It has also been shown in~\cite{aubrun2023annihilating} that if the central map $1\oplus u$ is $\infty$-max-entanglement annihilating central map, then the map $u:X\ra Y$ satisfies  
\begin{equation}\label{equ:TPNorm}
\|u^{\otimes k}:X^{\otimes_\epsilon k}\ra Y^{\otimes_\pi k}\|\leq 1 ,
\end{equation}
for every $k\in\N$. Here, $X^{\otimes_\epsilon k}$ denotes the $k$-fold injective tensor product of the space $X$ with itself, and $Y^{\otimes_\pi k}$ denotes the $k$-fold projective tensor product of the space $Y$ with itself. It has been shown in~\cite{aubrun2024limit} that having \eqref{equ:TPNorm} for every $k\in\N$ implies that $\gamma^*_2(u)\leq 1$. Since any $\infty$-max-entanglement annihilating map also breaks entanglement with Lorentz cones, it is a natural question how the condition $\gamma^*_2(u)\leq 1$ and the property of $1\oplus u$ being Lorentz entanglement breaking are related. Theorem \ref{thm:Main1} answers this question by showing that these two properties are equivalent. We have summarized the various implications and contributions in the following figure:\hfill \\

\begin{tikzpicture}
\node[] (A) at (0,0) {\footnotesize $\forall k\in\N, (1\oplus u)^{\otimes k}\lb \cC^{\otimes_{\max} k}_X\rb\subseteq \cC^{\otimes_{\min} k}_Y$};
\node[] (B) at (7,0) {\footnotesize $\sup_{k\in\N} \|u^{\otimes k}:X^{\otimes_{\epsilon} k}\ra Y^{\otimes_{\pi} k}\|\leq 1$};
\node[] (C) at (0,-2) {\footnotesize $1\oplus u\in \LorEB(\cC_X,\cC_Y)$};
\node[] (D) at (7,-2) {\footnotesize $\gamma^*_2(u)\leq 1$};
\node[] (F) at (3.5,-0.5) {\footnotesize \cite[Theorem 6.5.]{aubrun2023annihilating}};
\node[] (G) at (-1.8,-1) {\footnotesize \cite[Corollary 5.5.]{aubrun2023annihilating}};
\node[] (H) at (8.8,-1) {\footnotesize \cite[Theorem 3.6.]{aubrun2024limit}};
\node[] (J) at (4,-2.4) {\footnotesize Theorem \ref{thm:Main1}};

\draw[-implies, double=white, double distance=3pt] (A.east) to[out=0, in=180, looseness=1] (B.west);
\draw[-implies, double=white, double distance=3pt] (A.south) to[out=270, in=90, looseness=1] (C.north);
\draw[implies-implies, double=white, double distance=3pt] (B.south) to[out=270, in=90, looseness=1] (D.north);
\draw[blue, implies-implies, double=white, double distance=3pt] (C.east) to[out=0, in=180, looseness=1] (D.west);

\end{tikzpicture}

\section{Preliminaries}

Throughout this article all vector spaces are over the reals and finite-dimensional. We start by reviewing some notions from the theory of Banach spaces and operator ideal norms.

\subsection{Operator ideals} Let $X$ and $Y$ denote normed spaces and $u:X\ra Y$ a linear operator. We denote the operator norm of $u$ by $\|u:X\ra Y\|$ or simply by $\|u\|$ if it is clear from context which spaces are considered. The \emph{nuclear norm} of $u$ is given by
\[
\Nuc(u) = \inf\lset \sum^k_{i=1} \|y_i\|\|x^*_i\|~:~k\in\N, y_i\in Y, x^*_i\in X^*,u=\sum^k_{i=1} y_ix^*_i \rset .
\]
When $X$ and $Y$ are Euclidean spaces the nuclear norm equals the trace norm. 

The operator norm and the nuclear norm are examples of so-called \emph{operator ideal norms}~\cite{tomczak1989banach}. In general, a norm $\alpha$ on the space $L(X,Y)$ (of linear operators between finite-dimensional normed spaces $X$ and $Y$) is called an operator ideal norm if $\alpha(u)=\|u\|$ whenever $u\in L(X,Y)$ is rank-$1$, and if $\alpha(w_2uw_1)\leq \|w_1\|\alpha(u)\|w_2\|$ holds for all $u\in L(X,Y)$, any finite-dimensional normed spaces $X',Y'$ and any linear operators $w_2:Y\ra Y'$, $w_1:X'\ra X$. For any operator ideal norm $\alpha$ on $L(X,Y)$, we may define the \emph{dual ideal norm} $\alpha^*$ on $L(Y,X)$ by
\[
\alpha^*(v) = \max\lset |\Tr\lbr uv\rbr|~:~u:X\ra Y, \alpha(u)\leq 1\rset .
\]
It is easy to show that $\alpha^*$ is itself an operator ideal norm and it is sometimes referred to as the trace-dual of $\alpha$. Examples of operator ideal norms include the \emph{Hilbert-space factorization norm} $\gamma_2$, its trace-dual the \emph{$2$-dominated norm} $\gamma^*_2$, and the \emph{$2$-summing norm} $\pi_2$, which were all introduced above. For more information on operator ideals see~\cite{tomczak1989banach}.

We will need some properties of the $2$-summing norm $\pi_2$. It is well-known that the $2$-summing norm is selfdual and that it can be characterized via so-called $2$-nuclear factorizations (see~\cite[Proposition 9.10.]{tomczak1989banach}). For any linear operator $u:X\ra Y$ we have  
\begin{equation}\label{equ:2NuclearFact}
\pi_2(u) = \inf \|v_2\|\|\Delta\|\|v_1\| ,
\end{equation}
where the infimum is over factorizations $\phi=v_2\Delta v_1$ for linear operators $v_1:X\ra \ell^k_\infty$ and $v_2:\ell^k_2\ra Y$, and diagonal operators $\Delta:\ell^k_\infty\ra \ell^k_2$ with $k\in\N$. When $X$ and $Y$ are Euclidean spaces, the $2$-summing norm $\pi_2(u:X\ra Y)$ equals the Hilbert-Schmidt norm $\text{hs}(u)$, see~\cite[Proposition 10.1.]{tomczak1989banach}. The $2$-summing norm of $u:X\ra Y$ can be expressed as 
\[
\pi_2(u) = \sup\lset \pi_2(uv)~:~v:\ell^n_2\ra X, \|v\|\leq 1\rset ,
\]
where $n$ can be chosen as the rank of the map $u$ (see~\cite[Proposition 9.7.]{tomczak1989banach}). 

\subsection{Mapping ideals and trace duality}

Analogous to the theory of Banach spaces, we can consider ideals of positive maps between ordered vector spaces. A \emph{mapping ideal} $\cI$ assigns to each pair of proper cones $\cC_A\subseteq \cV_A$ and $\cC_B\subseteq \cV_B$ a convex cone $\cI(\cC_A,\cC_B)\subseteq L(\cV_A,\cV_B)$ satisfying the following conditions:
\begin{enumerate}
\item The cones $\cI(\cC_A,\cC_B)$ are closed.
\item We have $\EB(\cC_A,\cC_B)\subseteq \cI(\cC_A,\cC_B)\subseteq \Pos(\cC_A,\cC_B)$.
\item For cones $\cC_{A'}\subseteq \cV_{A'}$ and $\cC_{B'}\subseteq \cV_{B'}$ and positive maps $A\in \Pos(\cC_{A'},\cC_A)$ and $B\in \Pos(\cC_{B},\cC_{B'})$ we have $BPA\in \cI\lb \cC_{A'},\cC_{B'}\rb$ whenever $P\in \cI\lb \cC_{A},\cC_{B}\rb$.
\end{enumerate}  

In particular, the cones $\cI(\cC_A,\cC_B)$ coming from a mapping ideal are proper cones in $L(\cV_A,\cV_B)$ for any pair of proper cones $\cC_A\subseteq \cV_A$ and $\cC_B\subseteq \cV_B$. Mapping ideals are the natural analogue in the category of ordered vector spaces of operator ideals studied in the category of Banach spaces (see, e.g.,~\cite{tomczak1989banach}). Similar concepts have been studied intensively in other settings (see, e.g., the mapping cones in~\cite{stormer2013positive}). 

Given a cone $\cM\subseteq \Pos(\cC_A,\cC_B)$ of positive maps, we define its trace-dual cone $\cM^*$ as the set of all linear maps $Q\in \cL(\cV_B,\cV_A)$ such that 
\[
\Tr\lbr PQ\rbr \geq 0 ,
\]
for all $P\in \cM$. It is easy to see that $\EB(\cC_B,\cC_A) = \Pos(\cC_A,\cC_B)^*$. The notion of trace-duality is equivalent to a common notion of duality based on the generalized Choi-Jamiolkowski isomorphism (see~\cite{choi1975completely,jamiolkowski1972linear}, and \cite{aubrun2023annihilating}) between $\cV_A^*\otimes \cV_B$ and $L(\cV_A,\cV_B)$. For a finite-dimensional normed space $\cV$, let $\hat{I}_\cV\in \cV^*\otimes \cV$ denote the tensor corresponding to the identity map $\ident_\cV:\cV\ra \cV$, and note that $\hat{I}_\cV = \sum^d_{i=1} e^*_i\otimes e_i$ using an Auerbach basis $\lset (e_i,e^*_i)\rset^d_{i=1}$ for $\cV$. Then, we have 
\begin{equation}\label{equ:ChoiJamiolkowski}
\Tr\lbr wu\rbr = \braket{(w\otimes \ident_{\cV_B^*})(\hat{I}_{\cV_B^*})}{(\ident_{\cV_A^*}\otimes u)(\hat{I}_{\cV_A})} = \braket{\hat{I}_{\cV_B^*}}{(w^*\otimes u)(\hat{I}_{\cV_A})},
\end{equation}
for all $u:\cV_A\ra \cV_B$ and any $w:\cV_B\ra \cV_A$, where we identify $\cV^{**}_B$ and $\cV_B$.

For a linear map $P:\cV_A\ra \cV_B$ we denote its \emph{adjoint} by $P^*:\cV^*_B\ra \cV^*_A$. Most mapping ideals we consider will be closed under taking adjoints, e.g., we have $P\in \Pos(\cC_A,\cC_B)$ if and only if $P^*\in \Pos(\cC^*_B,\cC^*_A)$. We say that a map $P:\cV_A\ra \cV_B$ belongs to the interior of $\Pos(\cC_A,\cC_B)$ if for every $L\in L(\cV_A,\cV_B)$ there is an $\epsilon>0$ such that $P+\epsilon L\in \Pos(\cC_A,\cC_B)$. We write $P\in\text{int}\lb\Pos(\cC_A,\cC_B)\rb$ in this case. It is easy to see that a map $P:\cV_A\ra \cV_B$ belongs to the interior of $\Pos(\cC_A,\cC_B)$ if and only if $P(x)\in\text{int}\lb \cC_B\rb$ for every $x\in \cC_A$.

We will often study positive maps $P\in \Pos(\cC_X,\cC_Y)$ between cones over finite-dimensional normed spaces. We call such a map \emph{normalized} if $P(e^X_0)=e^Y_0$ where $e^X_0=(1,0)\in \R\times X$ and $e^Y_0=(1,0)\in \R\times Y$, and \emph{dual-normalized} if $P^*\in \Pos(\cC_{Y^*},\cC_{X^*})$ is normalized.

\subsection{Lorentz cones and Lorentz automorphisms}\label{sec:LorentzCones} 
For any $n\in\N$ we define the \emph{$(n+1)$-dimensional Lorentz cone} as 
\[
\cL_n = \lset (t,x)\in \R\times \ell^n_2~:~t\geq \|x\|_2\rset \subseteq \R^{n+1},
\]
where $\ell^n_2 = (\R^n,\|\cdot\|_2)$ is the $n$-dimensional real Euclidean space. For each $n\in\N$ the Lorentz cone $\cL_n$ is proper and an example of a so-called \emph{symmetric} cone. Recall that a closed cone $\cC$ is called symmetric if $\cC\simeq \cC^*$ and the automorphism group $\Aut(\cC)$ acts transitively on the interior $\text{int}\lb \cC\rb$, i.e., for any $x,y\in \text{int}\lb \cC\rb$ there is an $A\in \Aut(\cC)$ such that $A(x)=y$. We will often use the identification $\cL_n\simeq \cL^*_n$ implicitly. When it clarifies the discussion we will denote the adjoint of a map $P:\R^{n+1}\ra \R^{m+1}$ by $P^T$ instead of $P^*$. In particular, we have $P\in \Pos(\cL_n,\cL_m)$ if and only if $P^T\in \Pos(\cL_m,\cL_n)$. The symmetry of the Lorentz cones has the following consequence:

\begin{lem}\label{lem:Continuity}
Any map $P\in \text{int}\lb\Pos(\cL_k,\cC_X)\rb$ can be written as a composition $P=QA$, where $Q\in \Pos(\cL_k,\cC_X)$ is dual-normalized and $A\in \Aut(\cL_k)$ is a Lorentz-automorphism.
\end{lem}

\begin{proof}
Consider a $P\in \text{int}\lb\Pos(\cL_k,\cC_X)\rb$ and note that $P^*(e^{X^*}_0)\in \text{int}(\cL_k)$. Since $\cL_k$ is symmetric, there exists an automorphism $B\in \Aut(\cL_k)$ such that $B(P^*(e^{X^*}_0))=e_0 = (1,0)\in \text{int}\lb\cL_k\rb$. As a consequence we have $PB^* =: Q \in \Pos(\cL_k,\cC_{X})$, which is dual-normalized by construction. Finally, we find that $P=QA$ with $A=(B^{*})^{-1}$.
\end{proof}

The automorphism group of the Lorentz cone is up to scalar factors equal to the \emph{orthochronous Lorentz group} $O^+(1,n)$. Every automorphism $A\in \Aut(\cL_n)$ can be written as 
\begin{equation} \label{equ:LorAutRepre}
A=c(1\oplus u_1)P_\alpha (1\oplus u_2),
\end{equation}
with orthogonal operators $u_1,u_2\in O(n)$, a positive factor $c>0$, and a Lorentz boost $P_\alpha\in \Aut(\cL_n)$ of the form
\[
P_\alpha=\begin{pmatrix} \cosh(\alpha) & \sinh(\alpha) &  \\ \sinh(\alpha) & \cosh(\alpha) & \\  &  & \one_{n-1} \end{pmatrix},
\]
where $\alpha\in\R$ is some parameter~\cite[Theorem 6.5.]{gallier2020differential}. It is known (see for example~\cite[Proposition C.1.]{aubrun2017alice}) that the automorphisms $\Aut(\cL_n)$ together with the positive rank-$1$ maps generate all extremal rays of the cone $\Pos(\cL_n,\cL_n)$.

As is the case for the simplex cones and the cones of positive semidefinite matrices, there is a Sinkhorn-type~\cite{sinkhorn1964relationship} normal form of the maps in the interior of positive maps between Lorentz cones shown in~\cite[Theorem 3.4.]{hildebrand2011lmi}. Specifically, for any $P\in \text{int}\lb\Pos(\cL_n,\cL_m)\rb$ there are automorphisms $A\in \Aut(\cL_n)$ and $B\in \Aut(\cL_m)$ such that $BPA=1\oplus v$ for a diagonal contraction $v:\ell^n_2\ra \ell^m_2$. An immediate consequence of this theorem is the following useful proposition:

\begin{prop}\label{prop:denseSetLor}
Let $\mathcal{S}\subseteq \Pos(\cL_n,\cL_m)$ denote a closed set, that is closed under composition with automorphisms, and such that $P\in \mathcal{S}$ implies that $P+\epsilon e^m_0(e^n_0)^*\in\mathcal{S}$ for all $\epsilon>0$ small enough. Then, the set 
\[
\lset B(1\oplus v)A~:~A\in \Aut(\cL_n), B\in \Aut(\cL_m), 1\oplus v\in \mathcal{S}\text{ for diagonal }v:\ell^n_2\ra \ell^m_2\rset
\]
is dense in $\mathcal{S}$. In particular, the conclusion is true whenever $\mathcal{S}=\cI(\cL_n,\cL_m)$ for some mapping ideal $\cI$.
\end{prop}

We finish this section by discussing a useful property of cones obtained by intersecting Lorentz cones with subspaces. We say that a generating cone $\cC\subset \cV$ is a \emph{retract of a cone $\cC'\subset \cV'$} if there are maps $\alpha\in \Pos(\cC,\cC')$ and $\beta\in \Pos(\cC',\cC)$ such that $\ident_{\cV} = \beta\alpha$, i.e., the identity on $\cV$ factors through $\cC'$. It is easy to see that the cone $\cC=\lset \lambda x\rset$ is a retract of any cone $\cC'\subseteq \cV'$ not equal to $\lset 0\rset$. The Lorentz cones have the following retract property.

\begin{lem}\label{lem:RetractPropLor}
Consider a subspace $\cS\subseteq \R^{n+1}$ of dimension $\dim(\cS)=k$. Then, the cone $\cL_n\cap \cS$ is a retract of the Lorentz cone $\cL_{k-1}$.
\end{lem}

\begin{proof}
It is enough to show the statement for $\dim(\cS)=n$, i.e., in the case where $S$ is a hyperplane in $\R^{n+1}$. The general statement follows since any subspace can be expressed as an intersection of hyperplanes.

If $\cL_n\cap \cS = \lset \lambda x~:~\lambda\in\R^+\rset$ for some $x\in \cL_n$, then clearly it is a retract of $\cL_{k-1}$. Consider the case, where $\cS$ intersects the interior of $\cL_n$, i.e., the cone $\cL_n\cap \cS$ is generating in $\cS$. Let $\cH=\lset (1,x)~:~x\in \R^n\rset\subseteq \R^{n+1}$ denote the hyperplane cutting out the standard bases $\cL_n\cap \cH$ of the Lorentz cone. Clearly, the base $\cL_n\cap \cH$ is an $n$-dimensional ball. A base of the cone $\cL_n\cap S$ is given by $(\cL_n\cap \cS)\cap \cH$, which is an $(n-1)$-dimensional ball. Let us denote by $x_0\in (\cL_n\cap \cS)\cap \cH$ the center of this ball, and extend $x_0$ to an orthogonal basis $x_0,x_1,\ldots ,x_{n-1}$ of $\cS$ such that $x_0+\sum^{n-1}_{i=1} \lambda_i x_i\in (\cL_n\cap \cS)\cap \cH$ if and only if $\|\lambda\|_2\leq 1$. Let us define a map $\alpha:\cS\ra \R^{n}$ by setting $\alpha(x_i)=e_i$ for every $i\in\lset 0,\ldots ,n-1\rset$ and extending linearily. Clearly, $\alpha\lb(\cL_n\cap \cS)\cap \cH\rb = \cL_{n-1}$. Furthermore, we define a map $\beta:\R^{n}\ra \cS$ by $\beta(e_i)=x_i$ extended linearily. Again, we have $\beta(\cL_{n-1})=(\cL_n\cap \cS)\cap \cH$. Finally, we note that $\beta\alpha = \ident_\cS$ and we conclude that $\cL_n\cap \cS$ is a retract of $\cL_{n-1}$.
\end{proof}

\section{Lorentz entanglement annihilation}

\subsection{Maps between Lorentz cones} For $n,m\in\N$, let us consider positive maps $P:\R^{n+1}\ra \R^{m+1}$ such that $(P\otimes P)\lb \cL_n\otimes_{\max} \cL_n\rb\subseteq \cL_m\otimes_{\min} \cL_m$. These maps are the max-entanglement annihilating maps between the Lorentz cones $\cL_n$ and $\cL_m$, and we will use the notation $\maxEA_2(\cL_n,\cL_m)$ to refer to the set of all such maps. Clearly, we have $\maxEA_2(\cL_n,\cL_m)=\LorEA_2(\cL_n,\cL_m)$. In the following, we will show that the set $\maxEA_2(\cL_n,\cL_m)$ is a closed convex cone and that its central maps are characterized by the Hilbert-Schmidt norm. We start with a lemma:

\begin{lem}\label{lem:ProdCentralHSLor}
For $n,m\in\N$ consider $v_1,v_2:\ell^n_2\ra \ell^m_2$ and $r,s\in \R_0^+$ such that $\hs(v_1)\leq s$ and $\hs(v_2)\leq t$. Then, we have
\[
((s\oplus v_1)\otimes (t\oplus v_2)) \lb \cL_n\otimes_{\max} \cL_n\rb\subseteq \cL_m\otimes_{\min} \cL_m .
\] 
\end{lem}

\begin{proof}
Consider $v_1,v_2:\ell^n_2\ra \ell^m_2$ and $r,s\in \R_0^+$ such that $\hs(v_1)\leq s$ and $\hs(v_2)\leq t$. We need to show that $\braket{w}{((s\oplus v_1)\otimes (t\oplus v_2))(z)}\geq 0$ for all $z\in \cL_n\otimes_{\max} \cL_n$ and all $w\in \cL_m\otimes_{\max} \cL_m$. Note that any $z\in \cL_n\otimes_{\max} \cL_n$ can be written as $(\ident_{n+1}\otimes Q)(\hat{I}_n)$ for some positive map $Q\in \Pos(\cL_n,\cL_n)$ and the identity tensor $\hat{I}_n=\sum^n_{i=0}e_i\otimes e_i$, and we have a similar representation of $w\in \cL_m\otimes_{\max} \cL_m$ with a positive map $P\in \Pos(\cL_m,\cL_m)$. By \eqref{equ:ChoiJamiolkowski} it suffices to show
\begin{equation}\label{equ:PosForm}
\Tr\lbr (s\oplus v_1^T)P(t\oplus v_2)Q\rbr\geq 0 ,
\end{equation}
for all positive maps $Q\in \Pos(\cL_n,\cL_n)$ and $P\in \Pos(\cL_m,\cL_m)$. Moreover, it is enough to show \eqref{equ:PosForm} for all extremal maps $P$ and $Q$, which are not entanglement-breaking. These maps are the Lorentz automorphism~\cite[Proposition C.1.]{aubrun2017alice}, and they can be written as 
\[
P=c_1(1\oplus u_1)P_\alpha (1\oplus u_2) \quad\text{ and }\quad Q=c_2(1\oplus u_3)Q_\beta (1\oplus u_4),
\]
with orthogonal operators $u_1,u_2\in O(m)$ and $u_3,u_4\in O(n)$, strictly positive factors $c_1,c_2>0$, and positive maps $P_\alpha\in \Pos(\cL_m,\cL_m)$ and $Q_\beta\in \Pos(\cL_n,\cL_n)$ of the form
\[
P_\alpha=\begin{pmatrix} \cosh(\alpha) & \sinh(\alpha) &  \\ \sinh(\alpha) & \cosh(\alpha) & \\  &  & \one_{m-1} \end{pmatrix},
\]
and 
\[
Q_\beta=\begin{pmatrix} \cosh(\beta) & \sinh(\beta) &   \\ \sinh(\beta) & \cosh(\beta) &  \\  &  & \one_{n-1}  \end{pmatrix},
\]
with parameters $\alpha,\beta\in\R$ (see~\cite[Theorem 6.5.]{gallier2020differential}). Since the Hilbert-Schmidt norm is invariant under composition with orthogonal transformations, it suffices to prove \eqref{equ:PosForm} for $P=P_\alpha$ and $Q=Q_\beta$ for all $\alpha,\beta\in\R$. Now, we write
\[
v_1 = \begin{pmatrix} r_1 & x^T_1\\ y_1 & V_1\end{pmatrix} ,
\] 
and 
\[
v_2= \begin{pmatrix} r_2 & x^T_2\\ y_2 & V_2\end{pmatrix} ,
\]
for $r_1,r_2\in\R$, $x_1,x_2\in \R^{n-1}$, $y_1,y_2\in \R^{m-1}$ and $V_1,V_2:\R^{n-1}\ra \R^{m-1}$. We find that \eqref{equ:PosForm} with $P=P_\alpha$ and $Q=Q_\beta$ is equivalent to
\begin{equation}\label{equ:PosForm2}
(s,z_1)M_{\alpha,\beta}\begin{pmatrix} t\\ z_2\end{pmatrix}\geq 0,
\end{equation}
where we used the matrix
\[
M_{\alpha,\beta} = \begin{pmatrix} \cosh(\alpha)\cosh(\beta) & \sinh(\alpha)\sinh(\beta) & 0 & 0 & 0 \\ \sinh(\alpha)\sinh(\beta) & \cosh(\alpha)\cosh(\beta) & 0 & 0 & 0 \\ 0 & 0 & \cosh(\beta) & 0& 0 \\ 0 & 0 & 0 & \cosh(\alpha) & 0 \\ 0 & 0 & 0 & 0 & 1\end{pmatrix} .
\]
and vectors $z_1,z_2\in\R^4$ given by 
\[
z^T_1=\begin{pmatrix} r_1,\|x_1\|_2, \|y_1\|_2 ,\hs(V_1)\end{pmatrix}
\]
and
\[
z^T_2=\begin{pmatrix} r_2,\frac{\braket{x_1}{x_2}}{\|x_1\|_2}, \frac{\braket{y_1}{y_2}}{\|y_1\|_2} ,\frac{\Tr\lbr V^T_1 V_2\rbr}{\hs(V_1)}\end{pmatrix}.
\]
Using that $\hs(v_1)\leq s$ and $\hs(v_2)\leq t$, it is easy to check that $(s,z_1),(t,z_2)\in \cL_{4}$. Since $M_{\alpha,\beta}^TJM_{\alpha,\beta}\geq (\cosh(\alpha)^2 + \cosh(\beta)^2 - 1)J$ for the diagonal matrix $J$ with diagonal entries $(1,-1,-1,-1,-1)$ and $\cosh(\alpha)\cosh(\beta)\geq \sinh(\alpha)\sinh(\beta)$ for all $\alpha,\beta$, we conclude that $M_{\alpha,\beta}\in\Pos(\cL_4,\cL_4)$. This implies \eqref{equ:PosForm2} and hence \eqref{equ:PosForm} finishing the proof. 
\end{proof}

As a consequence, we obtain the following: 

\begin{thm}\label{thm:HSAndEALor}
Let $v:\ell^n_2\ra \ell^m_2$ be a linear map and $t\in\R$. We have $t\oplus v\in \maxEA_2(\cL_n,\cL_m)$ if and only if $\hs(v)\leq t$. 
\end{thm}

\begin{proof}
If $\hs(v)\leq t$, then we conclude by Lemma \ref{lem:ProdCentralHSLor} that $t\oplus v\in \maxEA_2(\cL_n,\cL_m)$. For the other direction assume that $t\oplus v\in \maxEA_2(\cL_n,\cL_m)$. This implies that 
\[
\braket{\hat{J}_m}{\lb(t\oplus v)\otimes (t\oplus v)\rb\lb \hat{I}_n\rb} = t^2 - \Tr\lbr v^Tv\rbr \geq 0 ,
\]
where $\hat{I}_n$ is the identity tensor and $\hat{J}_m\in \cL_m\otimes_{\max}\cL_m$ is the tensor given by $e_0\otimes e_0-\sum^m_{i=1} e_i\otimes e_i$. We conclude that $\hs(v)\leq t$.
\end{proof}

The special case of $n=m=3$ in Theorem \ref{thm:HSAndEALor} can also be obtained from the characterization of unital and trace-preserving $2$-locally entanglement annihilating maps on $M_2(\C)$ found in~\cite{filippov2012local}. This can be seen as follows: Using the Bloch ball representation (see, e.g.,~\cite[Section 2.1.2.]{aubrun2017alice}), the central maps in $\maxEA_2(\cL_3,\cL_3)$ correspond to the unital and trace-preserving positive maps $P:M_2(\C)_{sa}\ra M_2(\C)_{sa}$ satisfying 
\[
(P\otimes P)\lb M_2(\C)^+\otimes_{\max} M_2(\C)^+\rb \subseteq M_2(\C)^+\otimes_{\min} M_2(\C)^+ .
\]
By~\cite{stormer1963positive,woronowicz1976positive}, any $X\in M_2(\C)^+\otimes_{\max} M_2(\C)^+$ can be written as a sum $X=Y_1+(\ident_2\otimes \vartheta_2)(Y_2)$ with the transpose map $\vartheta_2:M_2(\C)_{sa}\ra M_2(\C)_{sa}$, given by $\vartheta_2(Y)=Y^T$ in any fixed basis, and positive semidefinite $Y_1,Y_2\in (M_2(\C)\otimes M_2(\C))^+$. This shows that the $2$-locally entanglement annihilating maps on $M_2(\C)$ characterized in~\cite{filippov2012local} coincide with the max-entanglement annihilating maps and the special case of $n=m=3$ in Theorem \ref{thm:HSAndEALor} follows.  

In general, the sets of max-entanglement annihilating maps between arbitrary pairs of proper cones are not convex (see Appendix~\ref{app:maxEANotConv}). However, the set of max-entanglement annihilating maps between Lorentz cones turns out to be convex. We start by the following observation.

\begin{thm}\label{thm:ProdPropEALor}
For any $P,Q\in \maxEA_2(\cL_n,\cL_m)$ we have
\[
(P\otimes Q)(\cL_n\otimes_{\max} \cL_n)\subseteq \cL_m\otimes_{\min} \cL_m .
\]
\end{thm}

\begin{proof}
By Theorem \ref{thm:HSAndEALor} and Lemma \ref{lem:ProdCentralHSLor} the statement is true for central maps $P,Q\in \maxEA_2(\cL_n,\cL_m)$. The general case follows from Proposition \ref{prop:denseSetLor}. 
\end{proof}

Next, we will prove some useful properties of the set $\maxEA_2(\cL_n,\cL_m)$. 

\begin{thm}\label{thm:PropsEALor}
For $n,m\in\N$ we have the following:
\begin{enumerate}
\item The set $\maxEA_2(\cL_n,\cL_m)$ is a closed convex cone.
\item If $P\in \maxEA_2(\cL_n,\cL_m)$, then $P^T\in \maxEA_2(\cL_m,\cL_n)$.
\item We have $\maxEA^*_2(\cL_n,\cL_m) = \maxEA_2(\cL_m,\cL_n)$.
\end{enumerate}
\end{thm}

\begin{proof}
It is clear that $\maxEA_2(\cL_n,\cL_m)$ is a closed set that is invariant under multiplication by positive scalars. Next, consider $P,Q\in \maxEA_2(\cL_n,\cL_m)$ and the convex combination $R = \lambda P + (1-\lambda)Q$ for $\lambda\in\lbr 0,1\rbr$. The tensor product $R\otimes R$ is a sum of the maps $P\otimes P$, $P\otimes Q$, $Q\otimes P$, and $Q\otimes Q$ with positive factors. By Theorem \ref{thm:PropsEALor} all those maps annihilate the entanglement in $\cL_n\otimes_{\max} \cL_n$ and we find that $R\in \maxEA_2(\cL_n,\cL_m)$. This shows that $\maxEA_2(\cL_n,\cL_m)$ is a convex cone. 

For the second statement consider $P\in \maxEA_2(\cL_n,\cL_m)$. By Proposition~\ref{prop:denseSetLor} we find $v_\epsilon:\ell^n_2\ra \ell^m_2$ for every $\epsilon >0$ satisfying $\hs(v_\epsilon)\leq 1$, and automorphisms $A_\epsilon\in \Aut(\cL_n)$ and $B_\epsilon\in \Aut(\cL_m)$ such that $P = \lim_{\epsilon\ra 0} B_\epsilon (1\oplus v_\epsilon) A_\epsilon$. As the Hilbert-Schmidt norm is preserved under taking adjoints, we conclude that $(B_\epsilon (1\oplus v_\epsilon) A_\epsilon)^T=A^T_\epsilon (1\oplus v^T_\epsilon) B^T_\epsilon\in \maxEA_2(\cL_m,\cL_n)$. Taking the limit this implies $P^T\in \maxEA_2(\cL_m,\cL_n)$.

To show the third statement, consider linear maps $P:\R^{m+1}\ra \R^{n+1}$ and $Q:\R^{n+1}\ra \R^{m+1}$. By \eqref{equ:ChoiJamiolkowski} we have
\[
\Tr\lbr PQ\rbr = \braket{\hat{I}_{n}}{(P\otimes Q^*)(\hat{I}_{m})},
\]
where $\hat{I}_n\in \R^{n+1}\otimes \R^{n+1}$ and $\hat{I}_{m}\in \R^{m+1}\otimes \R^{m+1}$ denote identity tensors. If $Q\in \maxEA_2(\cL_n,\cL_m)$, then $Q^*\in \maxEA_2(\cL_m,\cL_n)$ by the argument given above. In this case, Theorem \ref{thm:ProdPropEALor} implies that $(P\otimes Q^*)(\hat{I}_{m})\in \cL_n\otimes_{\min} \cL_n$ whenever $P\in \maxEA_2(\cL_m,\cL_n)$ since $\hat{I}_{m}\in \cL_m\otimes_{\max} \cL_m$. We conclude that $\Tr\lbr PQ\rbr\geq 0$ whenever $P\in \maxEA_2(\cL_m,\cL_n)$ and $Q\in \maxEA_2(\cL_n,\cL_m)$. Consider now a linear map $P:\R^{m+1}\ra \R^{n+1}$ such that $\Tr\lbr PQ\rbr\geq 0$ for any $Q\in \maxEA_2(\cL_n,\cL_m)$. Since $\EB(\cL_n,\cL_m)\subseteq\maxEA_2(\cL_n,\cL_m)$ we conclude that $P\in \Pos(\cL_m,\cL_n)$. Using the Sinkhorn type normal form~\cite[Theorem 3.4.]{hildebrand2011lmi}, we can find $v_\epsilon:\ell^n_2\ra \ell^m_2$ satisfying $\|v_\epsilon\|\leq 1$, and automorphisms $A_\epsilon\in \Aut(\cL_n)$ and $B_\epsilon\in \Aut(\cL_m)$ such that $P + \epsilon e^{m}_0(e^{n}_0)^* = B_\epsilon (1\oplus v_\epsilon) A_\epsilon$ for each $\epsilon>0$. Then, we have $\Tr\lbr (1\oplus v_\epsilon)Q\rbr\geq 0$ for every $Q\in \maxEA_2(\cL_n,\cL_m)$ and all $\epsilon>0$. Restricting to $Q=(1\oplus v')$ with $\hs(v')\leq 1$ and using selfduality of the Hilbert-Schmidt norm shows that $\hs(v_\epsilon)\leq 1$ for every $\epsilon>0$. We conclude that $P\in \maxEA_2(\cL_m,\cL_n)$ as it is a limit of maps in the closed set $\maxEA_2(\cL_m,\cL_n)$.  
\end{proof}

For any $n\in\N$ recall that $J_n\in \Pos(\cL_n,\cL_n)$ denotes the map such that $J_n(e_0)=e_0$ and $J_n(e_i)=-e_i$ for all $i\in\lset 1,\ldots ,n\rset$. The next theorem is useful for checking when positive maps between Lorentz cones are max-entanglement annihilating. \\

\begin{thm}\label{thm:maxEALorCriteria}
For $P\in \Pos(\cL_n,\cL_m)$ the following are equivalent:
\begin{enumerate}
\item We have $P\in\maxEA_2(\cL_n,\cL_m)$.
\item The eigenvalues $\lambda_0,\lambda_1,\ldots , \lambda_{m}$ of the matrix $J_mPJ_nP^T$ are non-negative and satisfy 
\[
\lambda_0 \geq \sum^m_{i=1} \lambda_i ,
\]
when ordered such that $\lambda_0\geq \max\lset \lambda_1,\cdots,\lambda_m\rset$.
\item We have $PJ_nP^T\in \EB(\cL_m,\cL_m)$.
\item We have $(P\otimes P)(\hat{J_n})\in \cL_m\otimes_{\min} \cL_m$. 
\end{enumerate}
\end{thm}

\begin{proof}
By Proposition \ref{prop:denseSetLor} and since all the statements in the theorem are preserved under taking limits along convergent sequences of maps, we may assume without loss of generality that $P=\alpha(1\oplus v)\beta$ with automorphisms $\alpha\in\Aut(\cL_m)$, $\beta\in\Aut(\cL_n)$ and a contraction $v:\ell^n_2\ra \ell^m_2$. 

To show that $(1)$ and $(2)$ are equivalent, we note that by~\cite[Lemma 3.8.]{hildebrand2011lmi} the eigenvalues of $J_mPJ_nP^T$ are real and coincide with the squares of singular values of the map $1\oplus v$. Since $v$ is a contraction, we have $\lambda_0=1$ and we have $\sum^m_{i=1} \lambda_i \leq \lambda_0$ if and only if $\hs(v)\leq 1$. By Theorem \ref{thm:HSAndEALor} this is equivalent to $P\in\maxEA_2(\cL_n,\cL_m)$. 

It is clear that $(3)$ and $(4)$ are equivalent and that $(1)$ implies $(4)$. Finally, since $\beta J_n\beta^T=J_n$ for every $\beta\in\Aut(\cL_n)$, we find that 
\[
PJ_nP^T = \alpha(1\oplus v)\beta J_n\beta^T(1\oplus v^T)\alpha^T = \alpha(1\oplus -vv^T)\alpha^T ,
\]
is entanglement breaking if and only if $\Nuc(vv^T)\leq 1$. Since $\Nuc(vv^T)=\hs(v)^2$, we can use Theorem \ref{thm:HSAndEALor} to conclude that $(3)$ implies $(1)$ finishing the proof.
 
\end{proof}

\subsection{Maps into Lorentz cones and the \texorpdfstring{$2$}{2}-summing norm}

We will now study the linear maps between cones $\cC_A$ and $\cC_B$ that annihilate Lorentz-entanglement, i.e., positive maps $P\in\Pos\lb \cC_A,\cC_B\rb$ such that
\[
(P\otimes P)\lb \cC_A\otimes_L \cC_A\rb\subseteq \cC_B\otimes_{\min} \cC_B .
\]
If $\cC_B=\cL_n$ for some $n\in\N$, then we can use the results from the previous section to understand the set $\LorEA_2(\cC_A,\cL_n)$ of these maps. We start with a proposition:

\begin{prop}\label{prop:CharactLorEAIntoLk}
For any proper cone $\cC\subseteq \cV$, any $k\in\N$ and any linear map $P:\cV\ra \R^{n+1}$ the following are equivalent:
\begin{enumerate}
\item We have $P\in \LorEA_2(\cC,\cL_n)$.
\item For every $k\in\N$ and every $Q\in \Pos(\cL_k,\cC)$ we have $PQ\in \maxEA_2(\cL_k,\cL_n)$.
\item For every $Q\in \Pos(\cL_{\dim(\cV)-1},\cC)$ we have $PQ\in \maxEA_2(\cL_{\dim(\cV)-1},\cL_n)$.
\end{enumerate}
\end{prop}

\begin{proof}
Using the definition of the Lorentzian tensor product and Theorem \ref{thm:ProdPropEALor} it is easy to see that the first two statements are equivalent. Clearly, the second statement implies the third. To see the remaining implication assume that for a linear map $P:\cV\ra \R^{n+1}$ there exists a $Q\in \Pos(\cL_k,\cC)$ for some $k\in\N$ with $k\geq \dim(\cV)$ such that $PQ\notin \maxEA_2(\cL_k,\cL_n)$. Then, by Theorem \ref{thm:PropsEALor}, we have $Q^*P^*\notin \maxEA_2(\cL_n,\cL_k)$. Consider the subspace $\cS = Q^*(\cV^*)$ satisfying $s:=\dim(\cS)\leq \dim(\cV)$. By Lemma \ref{lem:RetractPropLor} the cone $\cL_k\cap \cS$ is a retract of $\cL_{s-1}$ and consequently there are linear maps $\alpha:\cS\ra \R^{s}$ and $\beta:\R^s\ra \cS$ satisfying $\alpha(\cL_k\cap \cS)\subseteq \cL_{s-1}$ and $\beta(\cL_{s-1})\subseteq \cL_k\cap \cS$ and $\beta\alpha=\ident_{\cS}$. As $\beta\alpha Q^*P^* = Q^*P^*\notin \maxEA_2(\cL_n,\cL_k)$ we conclude that $\alpha Q^*P^*\notin \maxEA_2(\cL_n,\cL_{s-1})$. Since $\alpha Q^*\in \Pos(\cC^*,\cL_{s-1})$ we have $Q'=(\alpha Q^*)^*\in \Pos(\cL_{s-1},\cC)$ and $PQ'\notin \maxEA_2(\cL_{s-1},\cL_k)$ finishing the proof.

\end{proof}

Combining this characterization of $\LorEA_2(\cC,L_k)$ with Theorem \ref{thm:ProdPropEALor} and Theorem \ref{thm:PropsEALor} gives the following:

\begin{thm}
For any proper cone $\cC$ and any $n\in\N$ we have the following:
\begin{enumerate}
\item For any $P,Q\in \LorEA_2(\cC,\cL_n)$ we have
\[
(P\otimes Q)\lb \cC\otimes_L \cC\rb\subseteq \cL_n\otimes_{\min} \cL_n .
\]
\item The set $\LorEA_2(\cC,\cL_n)$ is a closed convex cone. 
\end{enumerate}
\end{thm}

Next, we aim to characterize the central maps in $\LorEA_2(\cC,L_k)$. For this we will need the following technical lemma:

\begin{lem}\label{lem:DiagonalPullbackMaxEA}
For any diagonal contraction $\Delta:\ell^k_\infty\ra \ell^k_2$ we have $(1\oplus \Delta)\in \LorEA_2(\cC_{\ell^k_\infty},\cL_{k})$. 
\end{lem}

\begin{proof}
By Proposition \ref{prop:CharactLorEAIntoLk} it suffices to show that $(1\oplus \Delta)P\in \maxEA_2(\cL_n,\cL_k)$ for any positive map $P\in \Pos(\cL_n,\cC_{\ell^k_\infty})$. By a continuity argument and Lemma \ref{lem:Continuity} we can restrict ourselves to dual-normalized $P\in \Pos(\cL_n,\cC_{\ell^k_\infty})$. Furthermore, it is enough to consider extremal dual-normalized maps $P\in \Pos(\cL_n,\cC_{\ell^k_\infty})$, which by Lemma \ref{lem:ExtremePointsTPLkCinfty} are of the form 
\[
P=\begin{pmatrix} 1 & 0 \\ s & 0 \\ 0 & A\end{pmatrix},
\]
up to a permutation of rows, for some partition $k=k_1+k_2$, a vector of signs $s\in\lset \pm 1\rset^{k_1}$ and a contraction $A:\ell^n_2\ra \ell^{k_2}_\infty$. Let us write 
\[
\Delta = \begin{pmatrix}\Delta_1 & 0 \\ 0 & \Delta_2 \end{pmatrix},
\]
for diagonal matrices $\Delta_1:\ell^{k_1}_\infty\ra \ell^{k_1}_2$ and $\Delta_2:\ell^{k_2}_\infty\ra \ell^{k_2}_2$ such that
\[
(1\oplus \Delta)P = \begin{pmatrix} 1 & 0 \\ \Delta_{1}s & 0 \\ 0 & \Delta_{2}A\end{pmatrix}.
\]
Let us set $\lambda=\sqrt{1-\|\Delta_{1}s\|^2_2}$ and note that since $\|\Delta\|=(\sum^{k}_{i=1} |\Delta_{ii}|^2)^{1/2}\leq 1$, we have $\|\Delta_{1}s\|^2_2 + \hs(\Delta_{2}A)^2\leq 1$, and hence $\lambda \geq \hs(\Delta_{2}A)$. Next, we introduce maps 
\[
T = (1\oplus \frac{1}{\lambda}\Delta_{2}A) \quad\text{ and }\quad S=\begin{pmatrix} 1 & 0 \\ \Delta_{1}s & 0 \\ 0 & \lambda \ident_{k}\end{pmatrix} ,
\] 
where $\ident_{k}:\ell^{k}_2\ra \ell^k_2$ denotes the identity. Since $\hs(\frac{1}{\lambda}\Delta_{2}A)\leq 1$, we find that $T\in \maxEA_2(\cL_{n}, \cL_{k})$ by Theorem \ref{thm:HSAndEALor}. We have $S\in\Pos(\cL_{k}, \cL_k)$ since 
\[
S\begin{pmatrix} 1 \\ x \end{pmatrix} = \begin{pmatrix} 1 \\ \Delta_{k_1}s \\ \lambda x \end{pmatrix},
\] 
and $\|\Delta_{1}s\|^2_2 + \lambda^2 \|x\|^2_2 \leq 1$ whenever $\|x\|_2\leq 1$. Since $(1\oplus \Delta)P = ST$ is a composition of the map $T\in \maxEA_2(\cL_{n}, \cL_{k})$ and the positive map $S\in\Pos(\cL_{k}, \cL_k)$, we conclude that it is max-entanglement annihilating as well.
\end{proof}

Now, we can characterize the central maps in $\LorEA_2(\cC_X,\cL_m)$.

\begin{thm}\label{thm:EquivalencesPi2}
For a linear map $v:X\ra \ell^m_2$ and $t\in\R$ we have $t\oplus v\in\LorEA_2(\cC_X,\cL_m)$ if and only if $\pi_2(v)\leq t$. 
\end{thm}

\begin{proof}
Without loss of generality we set $t=1$. Consider a linear map $v:X\ra \ell^m_2$ with $\pi_2(v)\leq 1$. Since the $2$-summing norm equals the $2$-nuclear norm (see~\cite[Proposition 9.10.]{tomczak1989banach}), we have a decomposition $v=u_2\Delta u_1$ with contractions $u_1:X\ra \ell^k_\infty$ and $u_2:\ell^k_2 \ra \ell^m_2$ and a diagonal contraction $\Delta:\ell^k_\infty\ra \ell^k_2$. By Lemma \ref{lem:DiagonalPullbackMaxEA} we have $(1\oplus \Delta)\in \LorEA_2(\cC_{\ell^k_\infty},\cL_{k})$. Since $1\oplus u_1\in \Pos(\cC_X,\cC_{\ell^k_\infty})$ and $1\oplus u_2\in \Pos(\cL_k,\cL_m)$, we conclude that $(1\oplus v)= (1\oplus u_2)(1\oplus \Delta)(1\oplus u_1)\in \LorEA_2(\cC_X,\cL_{m})$. 

For the other direction consider $v:X\ra \ell^m_2$ such that $(1\oplus v)\in\LorEA_2(\cC_X,\cL_m)$. By Proposition \ref{prop:CharactLorEAIntoLk} we find that $(1\oplus vu)\in \maxEA_2(\cL_n,\cL_m)$ and by Theorem \ref{thm:HSAndEALor} that $\hs(vu)\leq 1$ for any $u:\ell^n_2\ra X$ with $\|u\|\leq 1$ and any $n\in\N$. By~\cite[Proposition 9.7.]{tomczak1989banach} we have $\pi_2(v)=\sup\lset \hs(vu)~:~u:\ell^n_2\ra X, \|u\|\leq 1\rset\leq 1$ finishing the proof. 
\end{proof}

We finish this section by discussing the dual cone of $\LorEA_2(\cC,\cL_n)$ for any proper cone $\cC\subseteq \cV$.

\begin{thm}\label{thm:dualLorEACLn}
For any proper cone $\cC\subseteq \cV$ and any $n\in\N$ we have
\[
\LorEA_2(\cC,\cL_n)^* = \overline{\conv\lset AB~:~A\in \Pos(\cL_k,\cC), B\in \maxEA_2(\cL_n,\cL_k), k\in\N\rset}.
\]
\end{thm}

\begin{proof}
Let $\cK(\cL_n,\cC)$ denote the set on the right-hand side of the desired equation. It is easy to see that $\cK(\cL_n,\cC)$ is a closed convex cone. Consider first $P\in \LorEA_2(\cC,\cL_n)$ and $A\in \Pos(\cL_k,\cC)$ and $B\in \maxEA_2(\cL_n,\cL_k)$. By Proposition \ref{prop:CharactLorEAIntoLk} we have $PA\in\maxEA_2(\cL_k,\cL_n)$ and by Theorem \ref{thm:PropsEALor} we have $\Tr\lbr (PA)B\rbr\geq 0$. This proves the inclusion $\LorEA_2(\cC,\cL_n)^*\supseteq \cK(\cL_n,\cC)$. By the bipolar theorem the remaining inclusion is equivalent to $\LorEA_2(\cC,\cL_n)\supseteq \cK(\cL_n,\cC)^*$. Consider a linear map $P\in \cK(\cL_n,\cC)^*$. By definition, we have $\Tr\lbr PAB\rbr\geq 0$ for all $A\in \Pos(\cL_k,\cC)$ and $B\in \maxEA_2(\cL_n,\cL_k)$. Using Theorem \ref{thm:PropsEALor} we have $PA\in \maxEA_2(\cL_k,\cL_n)$ for any $A\in \Pos(\cL_k,\cC)$. By Proposition \ref{prop:CharactLorEAIntoLk} this shows that $P\in \LorEA_2(\cC,\cL_n)$. 
\end{proof}

The next proposition shows that a central map $1\oplus v$ for $v:\ell^m_2\ra X$ is in the dual cone of the Lorentz-entanglement annihilating maps $\LorEA_2(\cC_X,\cL_n)$ if and only if the $2$-summing norm satisfies $\pi_2(v)\leq 1$.

\begin{prop}\label{prop:EquivalencesPi22}
For $v:\ell^m_2\ra X$ and $t\in\R$ the following are equivalent:
\begin{enumerate}
\item We have $\pi_2(v)\leq t$.
\item We have $t\oplus v = PQ$ for some $P\in \Pos(\cL_k,\cC_X)$ and $Q\in \maxEA_2(\cL_m,\cL_k)$.
\item We have $t\oplus v \in \LorEA_2(\cC_X,\cL_n)^*$.
\end{enumerate}
\end{prop}

In particular, we have $t\oplus v\in \maxEA_2(\cL_m,\cC_X)$ whenever $\pi_2(v)\leq t$. 

\begin{proof}
Consider a linear map $v:\ell^m_2\ra X$ satisfying $\pi_2(v)\leq 1$. By~\cite[Proposition 9.10.]{tomczak1989banach} we have a factorization $v=u_2\Delta u_1$ with contractions $u_1:\ell^m_2\ra \ell^k_\infty$ and $u_2:\ell^{k}_2\ra X$ and a diagonal contraction $\Delta:\ell^k_\infty\ra \ell^k_2$. Clearly, we have $\pi_2(\Delta u_1)=\hs(\Delta u_1)\leq 1$ and by Theorem \ref{thm:HSAndEALor} we have $Q:=1\oplus\Delta u_1\in \maxEA_2(\cL_m,\cL_k)$. Now, we conclude that $1\oplus v = PQ$ where $P\in \Pos(\cL_k,\cC_X)$ and $Q\in \maxEA_2(\cL_m,\cL_k)$ showing that the first statement implies the second. By Theorem \ref{thm:dualLorEACLn}, the second statement implies the third. Finally, consider $v:\ell^m_2\ra X$ and $t\in\R^+$ such that $t\oplus v\in \LorEA_2(\cC_X,\cL_n)^*$. Using Theorem \ref{thm:EquivalencesPi2}, we find that $\Tr\lbr(t\oplus v)(1\oplus w)\rbr = t + \Tr\lbr vw\rbr\geq 0$ for all $w:X\ra \ell^m_2$ satisfying $\pi_2(w)\leq 1$. We conclude that $\pi_2(v)\leq t$ since the $2$-summing norm is selfdual (see~\cite[Proposition 9.10.]{tomczak1989banach}). 
\end{proof}

Finally, we can consider the general case and prove the following theorem: 

\begin{thm}\label{thm:2summingFullequ}
For a linear map $v:X\ra Y$ and $t\in\R$ the following are equivalent:
\begin{enumerate}
\item We have $\pi_2(v)\leq t$.
\item For any positive map $S\in \Pos(\cL_m,\cC_X)$ there exist a positive map $P\in \Pos(\cL_k,\cC_Y)$ and $Q\in \maxEA_2(\cL_m,\cL_k)$ such that $(1\oplus v)S = PQ$.
\end{enumerate}
\end{thm}

\begin{proof}
Without loss of generality we can restrict to $t=1$. To show that the first statement implies the second, consider $v:X\ra Y$ with $\pi_2(v)\leq 1$. By~\cite[Proposition 9.10.]{tomczak1989banach} we have a factorization $v=u_2\Delta u_1$ with contractions $u_1:\ell^m_2\ra \ell^k_\infty$ and $u_2:\ell^{k}_2\ra Y$ and a diagonal contraction $\Delta:\ell^k_\infty\ra \ell^k_2$. For any $S\in \Pos(\cL_m,\cC_X)$ the map $(1\oplus \Delta)(1\oplus u_1)S\in \maxEA_2(\cL_m,\cL_k)$ by Proposition \ref{prop:CharactLorEAIntoLk}, since $1\oplus \Delta\in \LorEA_2(\cC_{\ell^k_\infty},\cL_{k})$ and both $1\oplus u_1$ and $S$ are positive. Since $1\oplus u_2\in \Pos(\cL_k,\cC_Y)$ we have the desired decomposition. 

To see that the second statement implies the first, we consider a $v:X\ra Y$ such that $1\oplus v$ has the stated decomposition property. Choosing $S=1\oplus u$ for a contraction $u:\ell^k_2\ra X$ and using Theorem \ref{prop:EquivalencesPi22} shows that $\pi_2(vu)\leq 1$. Since $u$ was a general contraction, we have $\pi_2(v)\leq 1$ by~\cite[Proposition 9.7.]{tomczak1989banach}.
\end{proof}

It is natural to ask whether the $2$-summing norm can be used to characterized Lorentz-entanglement annihilating maps. The next proposition establishes one such implication.

\begin{prop}
For any linear map $v:X\ra Y$ and $t\in \R$ with $\pi_2(v)\leq t$ we have that $t\oplus v\in \LorEA_2(\cC_X,\cC_Y)$.
\end{prop}

\begin{proof}
For any $m\in\N$ consider $A,B\in \Pos(\cL_m,\cC_X)$. By Theorem \ref{thm:2summingFullequ} we have a factorizations $(t\oplus v)A=P_AQ_A$ and $(t\oplus v)B=P_BQ_B$ with $P_A,P_B\in \Pos(\cL_m,\cC_Y)$ and $Q_A,Q_B\in \maxEA_2(\cL_m,\cL_k)$. By Theorem \ref{thm:ProdPropEALor} we have 
\[
  \lb(t\oplus v)\otimes (t\oplus v)\rb\lb A\otimes B\rb(\hat{I}_m) = \lb P_A\otimes P_B\rb\lb Q_A\otimes Q_B\rb(\hat{I}_m) \in \cC_Y\otimes_{\min} \cC_Y.
\]
The proof is finished by recalling the definition of the Lorentzian tensor product. 
\end{proof}

\section{Lorentz-entanglement breaking maps}

As a direct consequence of~\cite[Lemma 5.2.]{aubrun2023annihilating} we have the following lemma characterizing the cone of Lorentz-entanglement breaking maps:

\begin{lem}
Let $\cC_A\subseteq \cV_A$ and $\cC_B\subseteq \cV_B$ denote proper cones and $P:\cV_A\ra \cV_B$ a linear map. The following are equivalent:
\begin{enumerate}
\item We have $P\in\LorEB(\cC_A,\cC_B)$.
\item For any $k,k'\in\N$ and positive maps $A\in\Pos(\cL_k,\cC_A)$ and $B\in\Pos(\cL_{k'},\cC_B)$ we have $BPA\in \EB(\cL_k,\cL_{k'})$.
\item For any $k\in\N$ and any positive map $A\in\Pos(\cL_k,\cC_A)$ we have $PA\in \EB(\cL_k,\cC_{B})$.
\item For any $k'\in\N$ and any positive map $B\in\Pos(\cC_{B},\cL_{k'})$ we have $BP\in \EB(\cC_{A},\cL_{k'})$.
\end{enumerate}
\end{lem}

It is now easy to see (compare also to~\cite[Theorem 5.3.]{aubrun2023annihilating}) that the trace-dual cone of $\LorEB(\cC_A,\cC_B)$ is given by 
\begin{equation}\label{equ:DualLorEB}
\LorEB(\cC_A,\cC_B)^* = \overline{\text{conv}\lset AB ~:~ B\in \Pos(\cC_B,\cL_k), A\in \Pos(\cL_k,\cC_A) , k\in\N\rset }.
\end{equation}
In the next section, we simplify this expression and show that it coincides with the set of Lorentz factorizable maps from \eqref{equ:LorFactDef}, i.e., the closure is not needed.

\subsection{Lorentz factorizable maps} We start with the following lemma:

\begin{lem}\label{lem:DimensionReduction}
Let $\cC_A\subseteq \cV_A$ and $\cC_B\subseteq \cV_B$ denote proper cones, and consider positive maps $A\in \Pos(\cC_A,\cL_n)$ and $B\in \Pos(\cL_n,\cC_B)$. Then, there exists a $k\leq \min(\dim(\cV_A),\dim(\cV_B))-1$ and positive maps $A'\in \Pos(\cC_A,\cL_k)$ and $B'\in \Pos(\cL_k,\cC_B)$ such that $BA=B'A'$.
\end{lem}

\begin{proof}
We will show that there exists a $k\leq \dim(\cV_A)-1$ with the desired property. The remaining case follows by taking duals of the maps in question using that the Lorentz cones are selfdual. Consider the subspace $\cS=A(\cV_A)\subseteq \R^{n+1}$ and note that $k:=\dim(\cS)-1\leq \dim(\cV_A)-1$. Consider the cone $\cS\cap \cL_n$, which by Lemma \ref{lem:RetractPropLor} is a retract of $\cL_k$. Hence, there are maps $\alpha:\cS\ra \R^{k+1}$ such that $\alpha(\cS\cap \cL_n)\subseteq \cL_k$ and $\beta:\R^{k+1}\ra \cS$ such that $\beta(\cL_k)\subseteq \cS\cap \cL_n$ satisfying $\ident_\cS=\beta\alpha$. We have
\[
BA = B\beta\alpha A = B'A' ,
\]
for $B'=B\beta$ and $A'=\alpha A$. The maps $A'$ and $B'$ have the desired properties.

\end{proof}

Now, we can show the following theorem:

\begin{thm}\label{thm:traceDualLorEBAndClosure}
For proper cones $\cC_A\subseteq \cV_A$ and $\cC_B\subseteq \cV_B$ we have 
\begin{align*}
\LorEB(\cC_B,\cC_A)^* &= \LorFact(\cC_A,\cC_B)\\
&=\conv\lset BA ~:~A\in \Pos(\cC_A,\cL_k), B\in \Pos(\cL_k,\cC_B)\rset ,
\end{align*}
where $k = \min(\dim(\cV_A),\dim(\cV_B))-1$. In particular, this set is closed.
\end{thm}

\begin{proof}
By Lemma \ref{lem:DimensionReduction} we immediately get the second equality in the statement of the theorem. By \eqref{equ:DualLorEB} and Lemma \ref{lem:DimensionReduction} we have 
\[
\LorEB(\cC_B,\cC_A)^* = \overline{\text{conv}\lset BA ~:~ A\in \Pos(\cC_A,\cL_k), B\in \Pos(\cL_k,\cC_B)\rset},
\]
for $k=\min(\dim(\cV_A),\dim(\cV_B))-1$. It remains to show that the closure is not needed. Fix $x_A\in \text{int}\lb\cC_A\rb$ and $\phi_B\in \text{int}\lb \cC^*_B\rb$, and define 
\[
\Pos_1\lb \cC_A,\cL_k\rb = \lset A\in \Pos\lb \cC_A,\cL_k\rb~:~e^*_0\lb A(x_A)\rb = 1\rset ,
\] 
where $e^*_0\in\text{int}\lb \cL^*_k\rb$ is the functional acting as $e^*_0(t,x)=t$ for $(t,x)\in \R\times \ell^k_2$, and 
\[
\Pos_1\lb \cL_k,\cC_B\rb = \lset B\in \Pos\lb \cL_k, \cC_B\rb~:~\phi_B\lb B(e_0)\rb = 1\rset .
\] 
Since $e^*_0\lb A(x_A)\rb = 0$ for $A\in \Pos\lb \cC_A,\cL_k\rb$ implies that $A=0$, we conclude that $\Pos_1\lb \cC_A,\cL_k\rb$ is a compact base of the cone $\Pos\lb \cC_A,\cL_k\rb$. Similarily, we have that $\Pos_1\lb \cL_k,\cC_B\rb$ is a compact base of the cone $\Pos\lb \cL_k, \cC_B\rb$. By Caratheodory's theorem we find that 
\[
K = \text{conv}\lset BA ~:~ A\in \Pos_1(\cC_A,\cL_k), B\in \Pos_1(\cL_k,\cC_B)\rset ,
\]  
is compact as well. The set $K$ is a compact base for the cone 
\[
\LorFact(\cC_A,\cC_B) = \text{conv}\lset BA ~:~ A\in \Pos(\cC_A,\cL_k), B\in \Pos(\cL_k,\cC_B)\rset, 
\]
and we conclude that it is closed.
\end{proof}

The following corollary is immediate using that the Lorentz factorizable maps are the dual cone of the Lorentz-entanglement breaking maps.

\begin{cor}\label{cor:LorEBCriteriaReducedDim}
For proper cones $\cC_A\subseteq \cV_A$ and $\cC_B\subseteq \cV_B$ and a linear map $P:\cV_A\ra \cV_B$ the following are equivalent:
\begin{enumerate}
\item We have $P\in\LorEB(\cC_A,\cC_B)$.
\item For $k=\min(\dim(\cV_A),\dim(\cV_B))-1$ and positive maps $A\in\Pos(\cL_k,\cC_A)$ and $B\in\Pos(\cL_{k},\cC_B)$ we have $BPA\in \EB(\cL_k,\cL_{k})$.
\end{enumerate}
\end{cor}

\subsection{Structure of central maps} Consider finite-dimensional normed spaces $X$ and $Y$. We aim to understand the central Lorentz-entanglement breaking maps between $\cC_X$ and $\cC_Y$. To do so, we start by proving a result about the structure of certain Lorentz factorizable maps from $\cC_{\ell^n_1}$ to $\cC_{\ell^m_\infty}$.

Let $u:\ell^m_1\ra \ell^n_\infty$ denote a linear map. By definition, we have $\gamma_2(u)\leq 1$ if and only if there exists $k\in\N$ and contractions $v_2:\ell^k_2\ra \ell^n_\infty$ and $v_1:\ell^m_1\ra \ell^k_2$ such that $u=v_2v_1$. We can specify the linear maps $v_1$ and $v_2$ by sets of unit vectors $x_1,\ldots , x_m, y_1,\ldots ,y_n\in B_{\ell^k_2}$ given by $x_j=v_1(e_j)$ for $j\in\lset 1,\ldots ,m\rset$ and $y_i=v_2^*(e_i)$ for $i\in\lset 1,\ldots ,n\rset$. This shows that $\gamma_2(v)\leq 1$ if and only if there exist unit vectors $x_1,\ldots , x_m, y_1,\ldots ,y_n\in B_{\ell^k_2}$ such that $u_{ij}=\braket{y_i}{x_j}$. We will use this well-known fact to prove the following lemma:

\begin{lem}\label{lem:LorentzFactorizableDiag}
For an extremal dual-normalized map $Q\in \Pos(\cL_k,\cC_{\ell^n_\infty})$ and an extremal positive map $P\in \Pos(\cL_k,\cC_{\ell^m_\infty})$ we have 
\[
QP^* = \begin{pmatrix} t & * \\ * & v\end{pmatrix} ,
\]
for some $t\in\R$ and $v:\ell^m_1\ra \ell^n_\infty$ satisfying $\gamma_2(v)\leq t$.
\end{lem}

\begin{proof}
Without loss of generality we may assume $t=1$ after suitable normalization, and we consider the following matrix forms 
\[
Q = \begin{pmatrix} 1 & 0 \\ s_1 & a^T_1 \\ \vdots & \vdots \\ s_n & a^T_n \end{pmatrix}, \quad P=\begin{pmatrix} 1 & w^T \\ r_1 & b^T_1 \\ \vdots & \vdots \\ r_m & b^T_m \end{pmatrix} ,
\]
with $s_1,\ldots ,s_n,r_1,\ldots ,r_m\in\R$ and $w,a_1,\ldots ,a_n, b_1,\ldots ,b_m\in \R^k$ with $\|w\|_2\leq 1$. By Lemma \ref{lem:ExtremePointsTPLkCinfty} we have $s_i\in\lset -1,0,1\rset$ and $\|a_i\|_2\leq 1$ for every $i\in\lset 1,\ldots ,n\rset$. Moreover, we have $a_i=0$ whenever $s_i\neq 0$. By Theorem \ref{thm:ExtrPosLkClinfty} we have either $r_j=\braket{b_j}{w}$ and $\|b_j\|_2 \leq 1$, or $r_j\in\lset -1,1\rset$ and $b_j=w$. Depending on the various cases we define vectors $x_1,\ldots ,x_m,y_1,\ldots ,y_n\in \R^{k+1}$ as follows: We set 
\begin{align*}
y^T_i &= \begin{cases} (0,a^T_i), &\text{ if }s_i=0 \\
s_i(\sqrt{1-\|w\|^2_2},w^T), &\text{ otherwise }, 
 \end{cases} \\
x_j &= \begin{cases} \pm(\sqrt{1-\|w\|^2_2},w^T), &\text{ if } (r_j,b^T_j)=\pm(1,w),\\
(0,b^T_j), &\text{ if }(r_j,b^T_j)=(\braket{b_j}{w},b^T_j) ,
 \end{cases}
\end{align*}
for $i\in\lset 1,\ldots ,n\rset$ and $j\in\lset 1,\ldots ,m\rset$. It is easy to verify that $\|y_i\|_2,\|x_j\|_2\leq 1$ and $v_{ij}=\braket{y_i}{x_j}$ showing that $\gamma_2(v)\leq 1$.

\end{proof}

We can now prove our main results. We start with the following theorem.

\begin{thm}\label{thm:LorEBviaGamma2Star}
For finite-dimensional normed spaces $X,Y$ and $v:X\ra Y$ the following are equivalent:
\begin{enumerate}
    \item We have $\gamma^*_2(v)\leq 1$.
    \item We have $1\oplus v\in \LorEB(\cC_X,\cC_Y)$.
\end{enumerate}
\end{thm}

\begin{proof}
Assume first that $1\oplus v\in \LorEB(\cC_X,\cC_Y)$. For $k\in\N$, consider contractions $u:\ell^k_2\ra X$ and $w:Y\ra \ell^k_2$ and note that $1\oplus wvu\in \EB(\cL_k,\cL_k)$. Then, we have 
\[
|\Tr\lbr vuw\rbr| \leq \text{Nuc}\lb wvu\rb \leq 1 .
\]
Since this holds for all contractions $u:\ell^k_2\ra X$ and $w:Y\ra \ell^k_2$, and every $k\in\N$, we conclude that $\gamma^*_2(v)\leq 1$.

For the other direction, let us assume first that $X=\ell^n_\infty$ and $Y=\ell^m_1$ for some $n,m\in\N$. Let $v:\ell^n_\infty\ra \ell^m_1$ satisfy $\gamma^*_2(v)\leq 1$. By a continuity argument and Lemma \ref{lem:Continuity} it suffices to show that $(1\oplus v)Q\in \EB(\cL_k,\cC_{\ell^m_1})$ for any dual-normalized $Q\in \Pos(\cL_k,\cC_{\ell^n_\infty})$. Using duality this is equivalent to
\[
\Tr\lbr P^*(1\oplus v)Q\rbr = \Tr\lbr (1\oplus v)QP^*\rbr \geq 0 ,
\]
for any $P\in \Pos(L_k,\cC_{\ell^m_\infty})$, which follows from Lemma \ref{lem:LorentzFactorizableDiag} by duality of $\gamma_2$ and $\gamma^*_2$. 

To show the case for general finite-dimensional normed spaces $X,Y$ note that $v:X\ra Y$ satisfies $\gamma^*_2(v)\leq 1$ if and only if it factorizes as $v=\beta v'\alpha$ where $\alpha:X\ra \ell^n_\infty$ and $\beta:\ell^m_1\ra Y$ are contractions and $v':\ell^n_\infty\ra \ell^m_1$ satisfies $\gamma^*_2(v')\leq 1$. Since $1\oplus v=(1\oplus \alpha)(1\oplus v')(1\oplus \beta)$ with positive maps $1\oplus \alpha$ and $1\oplus \beta$ and $1\oplus v'\in \LorEB(\cC_{\ell^n_\infty}, \cC_{\ell^m_1})$ we conclude that $1\oplus v\in \LorEB(\cC_{X}, \cC_{Y})$.
\end{proof}

As a corollary we have the following:

\begin{cor}\label{cor:gamma2Lor}
For finite-dimensional normed spaces $X,Y$ and $v:X\ra Y$ the following are equivalent:
\begin{enumerate}
    \item We have $\gamma_2(v)\leq 1$.
    \item We have $1\oplus v\in \LorFact(\cC_X,\cC_Y)$.
\end{enumerate}
\end{cor}

\begin{proof}
It is clear that the first statement implies the second. To show the other direction we consider $w:Y\ra X$ with $\gamma^*_2(w)\leq 1$. Since $1\oplus w\in \LorEB(\cC_Y,\cC_X)$ we find that
\[
1 + \Tr\lbr wv\rbr = \Tr\lbr (1\oplus w)(1\oplus v)\rbr \geq 0 .
\]
Since this holds for all $w:Y\ra X$ with $\gamma^*_2(w)\leq 1$ we conclude that $\gamma_2(v)\leq 1$.
\end{proof}

\subsection{Factorizable maps breaking Lorentz entanglement}

The $2$-dominated norm can be characterized via a factorization through Euclidean spaces, see~\eqref{equ:gamma2Star}. It is a natural question whether the Lorentz-entanglement breaking maps admit a similar characterization. We have the following theorem:

\begin{thm}\label{thm:LorEBviaFactorization}
For proper cones $\cC_A\subseteq \cV_A$ and $\cC_B\subseteq \cV_B$ and $n\in\N$ consider $A\in\LorEA(\cC_A,\cL_n)$ and $B\in\LorEA(\cC^*_B,\cL_n)$. Then, we have $P=B^*A\in\LorEB(\cC_A,\cC_B)$.
\end{thm}

\begin{proof}
Consider a map $P=B^*A\in\Pos(\cC_A,\cC_B)$ with $A\in\LorEA_2(\cC_A,\cL_n)$ and $B\in\LorEA_2(\cC^*_B,\cL_n)$. For any $Q_1\in \Pos(\cL_k,\cC_A)$ and $Q_2\in \Pos(\cC_B,\cL_k)$ we have 
\[
\Tr\lbr Q_2PQ_1\rbr = \Tr\lbr (BQ^*_2)^*(AQ_1)\rbr \geq 0, 
\]
using the third statement in Proposition \ref{thm:PropsEALor}, since $AQ_1\in \maxEA_2(\cL_k,\cL_n)$ and $(BQ^*_2)^*\in \maxEA_2(\cL_n,\cL_k)$ by the second statement in Proposition \ref{thm:PropsEALor}.
\end{proof}

\section{Specific cases and examples}

\subsection{The cone over the square} Let us consider first the cone $\cC_{\ell^2_\infty}$, i.e., the cone in $\R^3$ with a square base. The normed space $\ell^2_\infty$ has the $2$-summing property (see~\cite{arias1995banach}), i.e., for any linear map $v:\ell^2_\infty\ra \ell^m_2$ we have $\pi_2(v)=\|v\|$. Now, we show an analogue of the $2$-summing property in the setting of cones:

\begin{thm}
For every $m\in\N$ we have $\Pos(\cC_{\ell^2_\infty}, \cL_m) = \LorEA_2(\cC_{\ell^2_\infty}, \cL_m).$
\end{thm} 

\begin{proof}
One inclusion is clear. For the other inclusion consider a positive map $P\in \Pos(\cC_{\ell^2_\infty}, \cL_m)$. We will show that $PQ\in \maxEA_2(\cL_n,\cL_m)$ for every $Q\in \Pos(\cL_n,\cC_{\ell^2_\infty})$, which by Proposition \ref{prop:CharactLorEAIntoLk} implies that $P\in \LorEA_2(\cC_{\ell^2_\infty}, \cL_m)$. By a continuity argument and Lemma \ref{lem:Continuity} it is enough to consider dual-normalized $Q\in \Pos(\cL_n,\cC_{\ell^2_\infty})$ and we can restrict further to extreme points of the dual-normalized positive maps. By Lemma \ref{lem:ExtremePointsTPLkCinfty} these are of the form 
\[
Q = \begin{pmatrix} 1 & 0 \\ s & 0 \\ 0 & A\end{pmatrix} ,
\]
where $s\in\lset \pm 1\rset^{k_1}$ and $A:\ell^n_2\ra \ell^{k_2}_\infty$ an extremal contraction, and $k_1+k_2=2$. We have three cases: If $k_1=2$ and $k_2=0$, then the maps $Q$ are entanglement breaking and hence we have $PQ\in \maxEA_2(\cL_n,\cL_m)$. If $k_1=k_2=1$, then $A$ is a functional of norm $1$ and without loss of generality we can consider $s=1$. It can then be verified that 
\[
(Q\otimes Q)(\hat{J}_n) = \frac{1}{2}\begin{pmatrix} 1 \\ 1\\ -1\end{pmatrix}\otimes \begin{pmatrix} 1 \\ 1\\ 1\end{pmatrix} + \frac{1}{2}\begin{pmatrix} 1 \\ 1\\ 1\end{pmatrix}\otimes \begin{pmatrix} 1 \\ 1\\ -1\end{pmatrix} \in \cC_{\ell^2_\infty}\otimes_{\min} \cC_{\ell^2_\infty}.
\]
Hence, we conclude that $PQ\in \maxEA_2(\cL_n,\cL_m)$ by Theorem \ref{thm:maxEALorCriteria}. Finally, we have the case where $k_1=0$ and $k_2=2$. For any contraction $A:\ell^n_2\ra \ell^2_\infty$ we have $\pi_2(A^*:\ell^2_1\ra \ell^n_2)=\|A^*\|_{1\ra 2}\leq 1$ by the $2$-summing property since $\ell^2_1\simeq \ell^2_\infty$ isometrically. Therefore, we have $\gamma^*_2(AA^*)\leq \pi_2(A^*)^2\leq 1$. By Theorem \ref{thm:LorEBviaGamma2Star} we find that for $Q=1\oplus A$ we have $QJ_nQ^T = 1\oplus (-AA^*)\in \LorEB(\cC_{\ell^2_1},\cC_{\ell^2_\infty})$. Using Theorem \ref{thm:maxEALorCriteria} we conclude that $PQ\in \maxEA_2(\cL_n,\cL_m)$ since $PQJ_nQ^TP^T\in \EB(\cL_m,\cL_m)$.

\end{proof}

Using Theorem \ref{thm:LorEBviaFactorization} we immediately obtain the following corollary:

\begin{cor}\label{cor:LorEBLorFactCellinfty1}
We have $\LorFact(C_{\ell^2_\infty}, C_{\ell^2_1})\subseteq \LorEB(C_{\ell^2_\infty}, C_{\ell^2_1})$.
\end{cor}

By Corollary \ref{cor:LorEBLorFactCellinfty1} any map in $\LorFact(C_{\ell^2_\infty}, C_{\ell^2_1})$ breaks entanglement with Lorentz cones. As maps in $\LorFact(C_{\ell^2_\infty}, C_{\ell^2_1})$, by definition, factor through Lorentz cones, we have the following: 

\begin{cor}
We have $\LorFact(C_{\ell^2_\infty}, C_{\ell^2_1})\subseteq \maxEA_2(C_{\ell^2_\infty}, C_{\ell^2_1})$.
\end{cor}

\subsection{Cones of positive semidefinite matrices}\label{sec:PSDcones}

In this section, we aim to construct Lorentz-entanglement breaking maps between matrix algebras. One could be tempted to construct such a map from a Lorentz-entanglement breaking central map between cones $\cC_X$ and $\cC_Y$ over finite-dimensional normed spaces $X$ and $Y$. Unfortunately, the most immediate potential construction of this type will always result in an entanglement breaking map as shown in the following theorem.

\begin{thm}\label{thm:PSDFactorizationEB}
Let $X$ and $Y$ denote finite-dimensional normed spaces and $v:X\ra Y$ a linear operator with $\gamma^*_2(v)\leq \lambda$. For any positive maps $\alpha\in\Pos\lb M_n(\C)^+,\cC_{X}\rb$ and $\beta\in\Pos\lb \cC_{Y}, M_m(\C)^+\rb$ the composition $\beta(\lambda\oplus v)\alpha$ is entanglement breaking. 
\end{thm}

\begin{proof}
Without loss of generality we may restrict to $\lambda=1$. Combining \eqref{equ:gamma2Star} and \eqref{equ:2NuclearFact} shows that any $v:X\ra Y$ with $\gamma^*_2(v)\leq 1$ admits a factorization $v=w_2uw_1$ with contractions $w_1:X\ra \ell^{k_1}_\infty$ and $w_2:\ell^{k_2}_1\ra Y$ and a linear operator $u:\ell^{k_1}_\infty\ra \ell^{k_2}_1$ satisfying $\gamma^*_2(u)\leq 1$. We can therefore restrict to $X=\ell^{k_1}_\infty$ and $Y=\ell^{k_2}_1$. 

Consider positive maps $\alpha\in\Pos\lb  M_n(\C)^+,\cC_{\ell^{k_1}_\infty}\rb$ and $\beta\in\Pos\lb \cC_{\ell^{k_2}_1}, M_m(\C)^+\rb$ and note that there are collections of selfadjoint matrices $A_0,\ldots, A_{k_1}\in  M_n(\C)_{sa}$ and $B_0,\ldots, B_{k_2}\in  M_m(\C)_{sa}$ satisfying $A_0\pm A_i\geq 0$ for all $i\in\lset 1,\ldots ,k_1\rset$ and $B_0\pm B_j\geq 0$ for all $j\in\lset 1,\ldots ,k_2\rset$ such that
\[
\alpha(x) = \begin{pmatrix} \Tr\lbr A^T_0 x\rbr \\ \Tr\lbr A^T_1 x\rbr \\ \vdots \\ \Tr\lbr A^T_{k_1} x\rbr\end{pmatrix} \quad\text{ and }\quad \beta(y) = \sum^{k_2}_{j=0} y_i B_i ,
\]
for any $x\in  M_n(\C)_{sa}$ and any $y\in \R^{k_2+1}$. For a linear operator $v:\ell^{k_1}_\infty\ra \ell^{k_2}_1$ with $\gamma^*_2(v)\leq 1$ consider the positive map $P=\beta(1\oplus v)\alpha\in \Pos( M_n(\C)^+, M_m(\C)^+)$. To show that $P$ is entanglement breaking, we will first show that it is completely positive. By~\cite[Theorem 2]{choi1975completely} the map $P$ is completely positive if and only if the Choi operator
\[
C_P = A_0\otimes B_0 + \sum^{k_1}_{i=1}\sum^{k_2}_{j=1} v_{ij} A_i\otimes B_j ,
\]
is positive semidefinite. By a continuity argument, we may assume without loss of generality that $A_0$ and $B_0$ are positive definite. Then, the operator $C_P$ is positive semidefinite if and only if the operator   
\[
C' = \one_n\otimes \one_m + \sum^{k_1}_{i=1}\sum^{k_2}_{j=1} v_{ij} A'_i\otimes B'_j ,
\] 
is positive semidefinite for $A'_i=A^{-1/2}_0A_i A_0^{-1/2}$ and $B'_j=B^{-1/2}_0B_j B_0^{-1/2}$. Consider a normalized vector $\ket{\psi}\in \C^n\otimes \C^m$ and note that 
\[
\bra{\psi}C'\ket{\psi} = 1 + \sum^{k_1}_{i=1}\sum^{k_2}_{j=1} v_{ij}\braket{a_i}{b_j} ,
\] 
where $\ket{a_i}=(A'_i\otimes \one_m)\ket{\psi}\in \C^n\otimes \C^m$ and $\ket{b_j} = (\one_n\otimes B'_j)\ket{\psi}\in \C^n\otimes \C^m$. Since $A'_i$ and $B'_j$ are contractions, we have $\braket{a_i}{a_i},\braket{b_j}{b_j}\leq 1$ for all $i$ and $j$. Finally, note that the inner products $\braket{a_i}{b_j}$ are real for all $i$ and $j$. Therefore, we may introduce normalized real vectors 
\[
\ket{\tilde{a}_i} = \begin{pmatrix} \text{Re}\lb \ket{a_i}\rb \\ \text{Im}\lb \ket{a_i}\rb\end{pmatrix}, \quad\text{ and }\quad\ket{\tilde{b}_j} = \begin{pmatrix} \text{Re}\lb \ket{b_j}\rb \\ \text{Im}\lb \ket{b_j}\rb\end{pmatrix} ,
\]
satisfying $\braket{\tilde{a}_i}{\tilde{b}_j} = \braket{a_i}{b_j}$. Hence, the operator $M:\ell^{k_2}_1\ra \ell^{k_1}_\infty$ with $\braket{e_i}{M(e_j)}=\braket{a_i}{b_j}$ satisfies $\gamma_2(M)\leq 1$. By duality of the norms $\gamma_2$ and $\gamma^*_2$ we find that 
\[
|\sum^{k_1}_{i=1}\sum^{k_2}_{j=1} v_{ij}\braket{a_i}{b_j}|\leq \gamma^*_2(v)\leq 1 ,
\]
and hence that $\bra{\psi}C'\ket{\psi}\geq 0$. We conclude that $C'$ and thereby also $C_P$ is positive semidefinite. We have shown that $\beta^*(1\oplus v)\alpha$ is completely positive for any choice of positive maps $\alpha\in\Pos\lb  M_n(\C)^+,\cC_{\ell^{k_1}_\infty}\rb$ and $\beta\in\Pos\lb  M_m(\C)^+,\cC_{\ell^{k_2}_\infty}\rb$, and a linear operator $v:\ell^{k_1}_\infty\ra \ell^{k_2}_1$ with $\gamma^*_2(v)\leq 1$. Composing such a map with any positive map $W\in \Pos( M_m(\C)^+, M_n(\C)^+)$ leads to a map $\beta^*(1\oplus v)\alpha W = \beta^*(1\oplus v)\alpha'$ of the same form, where $\alpha'=\alpha W$. Since all such compositions are completely positive as well, we find that the maps of the considered form are entanglement breaking.
\end{proof}

We will now construct an explicit example of a map in $\LorEB\lb  M_3(\C)^+,  M_3(\C)^+\rb$ that is not entanglement breaking. 

\begin{example} We start by defining a completely positive map $T:M_3(\C)_{sa}\ra M_3(\C)_{sa}$ by setting 
\[
T = \frac{3257}{6884} \text{Ad}_{K_1} + \frac{450}{1721} \text{Ad}_{K_2} + \frac{450}{1721} \text{Ad}_{K_3} + \frac{27}{6884} \text{Ad}_{K_4},
\]
where $\text{Ad}_{K}$ denotes the linear map $X\mapsto KXK^\dagger$ and 
\begin{align*}
K_1 &= \frac{1}{\sqrt{2}}\begin{pmatrix} 1 & 0 & 0 \\ 0 & 1 & 0\\ 0 & 0 & 0\end{pmatrix} , \\ 
K_2 &= \frac{1}{2}\begin{pmatrix} 0 & \frac{1}{6}\sqrt{\frac{131}{2}} & 0 \\ \frac{1}{6}\sqrt{\frac{131}{2}} & 0 & -\frac{3}{5}\\ \frac{1}{30} & 0 & 0\end{pmatrix} ,\\
K_3 &= \frac{1}{2}\begin{pmatrix} \frac{1}{6}\sqrt{\frac{131}{2}} & 0 & \frac{3}{5} \\ 0 & -\frac{1}{6}\sqrt{\frac{131}{2}} & 0\\ 0 & \frac{1}{30} & 0\end{pmatrix}, \\
K_4 &= \frac{1}{\sqrt{3}}\begin{pmatrix} 0 & 1 & 0 \\ -1 & 0 & 0\\ 0 & 0 & 1\end{pmatrix}.
\end{align*}
These parameters and operators are chosen such that the Choi operator $C_T$ is the quantum state constructed in~\cite{vertesi2014disproving} providing a counterexample to Peres' conjecture. In particular, the map $T$ is completely positive and completely copositive, i.e., the composition $\vartheta_3T$ of $T$ with the transpose map $\vartheta_3:M_3(\C)_{sa}\ra M_3(\C)_{sa}$ (with respect to any fixed basis) is completely positive. Next, we consider the normalized vectors  
\begin{align*}
\ket{a_1} &= \frac{1}{5}\lb -\ket{0} + \sqrt{3}\ket{1} + \sqrt{21}\ket{2}\rb \\ 
\ket{a_2} &= \frac{1}{5}\lb 2\ket{0} + \sqrt{21}\ket{2}\rb ,\\
\ket{a_3} &= \frac{1}{5}\lb -\ket{0} - \sqrt{3} + \sqrt{21}\ket{2}\rb ,
\end{align*}
and define contractions $A_1,A_2,A_3\in M_3(\C)_{sa}$ by setting $A_i = 2\proj{a_i}{a_i}-\one_3$ for every $i\in\lset 1,2,3\rset$. This set of vectors is found in~\cite{vertesi2014disproving} and used to specify a particular collection of projective measurements exhibiting nonlocality of the quantum state $C_T$. The operators $A_i$ are the observables corresponding to those measurements. We use them to define a positive map $A\in \Pos(\cC_{\ell^3_1},M_3(\C)^+)$ by 
\[
A(x) = x_0 \one_3 + x_1 A_1 + x_2 A_2 + x_3 A_3 .
\]
Note that the composition $TA\in \Pos(\cC_{\ell^3_1},M_3(\C)^+)$ is not entanglement breaking since there is a map $B\in \Pos(M_3(\C)^+,\cC_{\ell^3_1})$ given by 
\[
B(X) = \Tr\lbr B_0 X\rbr e_0 + \Tr\lbr B_1 X\rbr e_1+ \Tr\lbr B_2 X\rbr e_2 + \Tr\lbr B_3 X\rbr e_3 ,
\]
with $B_0=\one_3$ and 
\begin{align*}
B_1 &= \begin{pmatrix} 1/2 & 28/97 & -28/97 \\ 
28/97 & 1/6 & -1/6 \\ 
-28/97 & -1/6 & -1/3\end{pmatrix} , \\
B_2 &= \begin{pmatrix} 0 & 0 & 0 \\ 
0 & 2/3 & 1/3 \\ 
0 & 1/3 & -1/3\end{pmatrix} ,\\
B_3 &= \begin{pmatrix} 1/2 & -28/97 & 28/97 \\ 
-28/97 & 1/6 & -1/6 \\ 
28/97 & -1/6 & -1/3\end{pmatrix} ,
\end{align*}
satisfying $\Tr\lbr BTA\rbr <0$. By cyclicity of the trace (and since entanglement breaking maps have non-negative trace) we find that $TAB\in\Pos\lb M_3(\C)^+,M_3(\C)^+\rb$ is not entanglement breaking. Finally, we prove that $TAB\in \LorEB(M_3(\C)^+,M_3(\C)^+)$ by showing that the map $TA\in \LorEB(\cC_{\ell^3_1},M_3(\C)^+)$. Identifying $\cL_3$ and $M_2(\C)^+$ using the Bloch ball representation (see, e.g.,~\cite[Section 2.1.2.]{aubrun2017alice}) and using Corollary \ref{cor:LorEBCriteriaReducedDim} it is enough to show that $QTA\in \EB(\cC_{\ell^3_1},M_2(\C)^+)$ for any $Q\in \Pos\lb M_3(\C)^+,M_2(\C)^+\rb$. By~\cite{woronowicz1976positive}, any $Q\in \Pos\lb M_3(\C)^+,M_2(\C)^+\rb$ can be written as a sum $Q=S_1 + \vartheta_2S_2$ with completely positive maps $S_1,S_2:M_3(\C)_{sa}\ra M_2(\C)_{sa}$ and the transpose map $\vartheta_2:M_2(\C)_{sa}\ra M_2(\C)_{sa}$ (in any fixed basis). Since $T:M_3(\C)_{sa}\ra M_3(\C)_{sa}$ is completely positive and completely copositive, it breaks entanglement with the cone $M_2(\C)^+$ as shown in~\cite[Theorem III.1.]{christandl2019composed}. We conclude that the compositions $S_1T$ and $S_2T$ are both entanglement breaking. Finally, this implies that the composition $QT$ is entanglement breaking for any $Q\in \Pos\lb M_3(\C)^+,M_2(\C)^+\rb$ and we conclude that $QTA\in \EB(\cC_{\ell^3_1},M_2(\C)^+)$. 
\end{example}

Since the map $P=TAB\in \LorEB(M_3(\C)^+,M_3(\C)^+)$ from the previous example has negative trace, we conclude by duality that the identity map $\ident_3\in \Pos(M_3(\C)^+,M_3(\C)^+)$ does not factor through a Lorentz cone (see~\eqref{equ:LorFactDef}). This could have also been deduced from results in~\cite{fawzi2019representing} albeit in a more involved way.

\section*{Acknowledgments}

We thank Guillaume Aubrun for pointing out the retraction property of Lorentz cones (Lemma~\ref{lem:RetractPropLor}) and that the closure is not needed in the Lorentzian tensor product. We thank Roy Araiza for helpful discussions and Emilie Mai Elkiær for helpful comments on ellipses. The authors acknowledge funding from The Research Council of Norway (project 324944). 

\section*{Data availability}

No data was used for the research described in the article.

\appendix

\section{Structure of \texorpdfstring{$\Pos(\cL_k,\cC_{\ell^n_\infty})$}{Pos(Lk,Cellinfty)}}

After normalizing the upper left corner, we may think of the maps $P\in \Pos(\cL_k,\cC_{\ell^n_\infty})$ as matrices
\begin{equation}\label{equ:Pform}
P=\begin{pmatrix} 1 & w^T \\ x & B\end{pmatrix} ,
\end{equation}
with $w\in \R^k$, $x\in \R^{n}$ and $B\in \R^{n\times k}$. We have $P\in\Pos(\cL_k,\cC_{\ell^n_\infty})$ if and only if $(1,w)\pm(x_i,b_i)\in \cL_k$ for every $i\in\lset 1,\ldots ,n\rset$ and where $b_i\in\R^k$ denote the rows of $B$. In particular, we have $1\pm x_i\geq 0$ for any $i\in\lset 1,\ldots ,n\rset$. Let us fix a vector $w\in \R^k$ and consider $\Pos_w(\cL_k,\cC_{\ell^n_\infty})$ to be the closed and convex set of positive maps with the specified vector $w$ appearing in the first row. Since $(1,w)\in \cL_k$ we have $\|w\|_2\leq 1$. We aim to find the extreme points of the set $\Pos_w(\cL_k,\cC_{\ell^n_\infty})$. 

Let us first consider the case where $\|w\|_2 = 1$ implying that $(1,w)$ is extremal in $\cL_k$. If $(x,b)$ is a row occuring in a matrix $P\in \Pos_w(\cL_k,\cC_{\ell^n_\infty})$, then we have $(1,w)\pm (x,b)\in \cL_k$. Since  
\[
(1,w) = \frac{1}{2}((1,w)+ (x,b)) + \frac{1}{2}((1,w)- (x,b)) ,
\]
we conclude that $(x,b)=\alpha(1,w)$ for some $\alpha\in \lbr -1,1\rbr$. If $P$ is extremal in $\Pos_w(\cL_k,\cC_{\ell^n_\infty})$, then we have $\alpha=1$ or $\alpha=-1$. We conclude that $P\in \Pos_w(\cL_k,\cC_{\ell^n_\infty})$ for $\|w\|_2=1$ is extremal if and only if $(x_i,b_i)\in\lset \pm(1,w)\rset$ for all $i\in\lset 1,\ldots ,n\rset$.

Let us now consider the case where $\|w\|_2<1$. We start by defining the set 
\[
\mathcal{K}_w = \lset (x,b)\in \R\times \R^k ~:~ 1\pm x \geq \|b\pm w\|_2\rset .
\]
Clearly, we have $P\in \text{Pos}_w(\cL_k,\cC_{\ell^n_\infty})$ if and only if it is of the form \eqref{equ:Pform} with rows $(x_i,b_i)\in \mathcal{K}_w$. To understand the structure of the set $\mathcal{K}_w$ we consider the two affine ice-cream cones 
\begin{align*}
\mathcal{A}_w &:= \lset (x,b)\in \R\times \R^k ~:~ 1 - x \geq \|b - w\|_2\rset , \\
\mathcal{B}_w &:= \lset (x,b)\in \R\times \R^k ~:~ 1 + x \geq \|b + w\|_2\rset.
\end{align*}
Clearly, $\mathcal{K}_w=\mathcal{A}_w\cap \mathcal{B}_w$. Consider the hyperplane $\cH\subset \R\times \R^k$ perpendicular to $(1,-w)$. We have the following lemma:

\begin{lem}\label{lem:intersectHyperplane}
$\mathcal{A}_w\cap \cH = \mathcal{B}_w\cap \cH = \mathcal{K}_w\cap \cH$.
\end{lem} 

\begin{proof}
For $(x,b)\in \mathcal{A}_w\cap \cH$ we have $1 - x \geq \|b - w\|_2$ and $x=\braket{b}{w}$. Using the Cauchy-Schwarz inequality we have
\[
\|w\|^2-\braket{b}{w} \leq \|w\| \|w-b\| \leq \|w\|(1-\braket{b}{w}),
\] 
which implies
\begin{equation}\label{equ:crucial}
0\leq (1-\|w\|)(\|w\|+\braket{b}{w})\leq (1-\|w\|)(1+\braket{b}{w}),
\end{equation}
where we used that $\|w\|<1$. We conclude that $1+x=1+\braket{b}{w}\geq 0$. Finally, we compute 
\[
\|b+w\|^2_2 = \|b-w\|^2_2 + 4\braket{b}{w} \leq (1-x)^2+4x = (1+x)^2.
\]
Since $1+x\geq 0$, we find that $(x,b)\in \mathcal{B}_w\cap \cH$. The other inclusion works the same way showing that $\mathcal{A}_w\cap \cH = \mathcal{B}_w\cap \cH$. Since $\mathcal{K}_w$ is the intersection of $\mathcal{A}_w$ and $\mathcal{B}_w$ the proof is finished.
\end{proof}

The first inequality in \eqref{equ:crucial} and its analogue for the other choice of sign imply in particular that $|\braket{b}{w}|\leq \|w\|_2$ whenever $(\braket{b}{w},b)\in \mathcal{K}_w\cap \cH$. A consequence of this is the following lemma. 

\begin{lem}\label{lem:normb}
If $b\in\R^k$ is such that $(\braket{b}{w},b)\in \mathcal{K}_w\cap \cH$, then we have $\|b\|_2\leq 1$.
\end{lem} 

\begin{proof}
Let $\alpha,\beta\in\R$ be such that 
\[
b = \alpha\frac{w}{\|w\|_2} + \beta v ,
\]
for some unit vector $v\in \R^k$ perpendicular to $w$. Note that $\braket{b}{w}=\alpha \|w\|_2$. By definition of $\mathcal{K}_w$ we have 
\[
\alpha^2 + \beta^2 + 2\alpha\|w\| + \|w\|^2 = \|b + w\|^2_2 \leq (1 + \alpha \|w\|)^2 = 1 + 2\alpha\|w\| + \alpha^2\|w\|^2 ,
\]
and hence 
\[
\|b\|^2_2 = \alpha^2+\beta^2 \leq 1 + (\alpha^2-1)\|w\|^2 .
\]
Since $|\braket{b}{w}|\leq \|w\|_2$ by the comment preceeding the statement of the lemma, we find that $|\alpha|=|\braket{b}{w}|/\|w\|_2 \leq 1$. We conclude that $\|b\|^2_2\leq 1$.
\end{proof}

We can now characterize the extreme points of $\mathcal{K}_w$.

\begin{lem}
For $\|w\|_2<1$, the extreme points of the closed convex set $\mathcal{K}_w$ are the two points $(1,w)$ and $(-1,-w)$, and the extreme points of $\mathcal{K}_w\cap \cH$. 
\end{lem}

\begin{proof}
The points $(1,w)$ and $(-1,-w)$ are the apex points of the affine cones $\mathcal{A}_w$ and $\mathcal{B}_w$, respectively. They are clearly contained in $\mathcal{K}_w$ and, since $\mathcal{K}_w=\mathcal{A}_w\cap \mathcal{B}_w$, they are extremal. Consider any point $(x,b)\in \mathcal{K}_w$. Let us assume that $x > \braket{b}{w}$. We may define a point $(x',b')\in \cH$ as the intersection of $\cH$ with the line from $(1,w)$ through $(x,b)$. Since $(x,b)$ and $(1,w)$ both belong to the affine cone $\mathcal{A}_w$ we conclude that $(x',b')\in \mathcal{A}_w\cap \cH = \mathcal{K}_w\cap \cH$ by Lemma \ref{lem:intersectHyperplane}. Now, we can express $(x,b)$ as a convex combination of $(1,w)$ and the point $(x',b')\in \mathcal{K}_w\cap \cH$. The case where $x < \braket{b}{w}$ works analogously, and we have shown that the extreme points of $\mathcal{K}_w$ different from $(1,w)$ and $(-1,-w)$ are contained in the set of extreme points of $\mathcal{K}_w\cap \cH$. Finally, we note that the line connecting $(1,w)$ and $(-1,-w)$ intersects $\cH$ in the point $(0,0)$, which is an interior point of $\mathcal{K}_w\cap \cH$ as $\|w\|_2<1$. Hence, all extremal points of $\mathcal{K}_w\cap \cH$ are also extremal in $\mathcal{K}_w$.
\end{proof}

Since $\mathcal{A}_w$ and $\mathcal{B}_w$ are affine ice-cream cones with $(0,0)$ in their respective interior, the convex set $\mathcal{K}_w\cap \cH= \mathcal{A}_w\cap \cH = \mathcal{B}_w\cap \cH$ is an ellipsoid. The following theorem identifies an important property of the extreme points of $\mathcal{K}_w$. 

\begin{thm}
For $\|w\|_2<1$, the extreme points of $\mathcal{K}_w$ are of the form
\[
\begin{cases} (1,w) \\ (-1,-w) \\ (\braket{b}{w},b) \end{cases} ,
\]
where in the last case $b\in\R^{k}$ satisfies $\|b\|_2\leq 1$.
\end{thm}

Using the connection to maps in $\text{Pos}_w(L_k,C_{\ell^n_\infty})$, we have the following:

\begin{thm}\label{thm:ExtrPosLkClinfty}
If the linear map $P:\R^{k+1}\ra \R^{n+1}$ is an extreme point of $\text{Pos}_w(\cL_k,\cC_{\ell^n_\infty})$, then either of the following conditions hold:
\begin{enumerate}
\item We have $\|w\|_2=1$ and there is vector $s\in \lset\pm 1\rset^n$ such that
\[
P = \begin{pmatrix} 1 & w^T \\ s & sw^T\end{pmatrix} .
\]
\item We have $\|w\|_2<1$, there is a partition $n=n_1+n_2$ for some $n_1,n_2\in\N$, a vector $s\in \lset \pm 1\rset^{n_1}$, and a contraction $B:\ell^k_2\ra \ell^{n_2}_\infty$ such that
\[
P = \begin{pmatrix} 1 & w^T \\ s & sw^T \\ Bw & B\end{pmatrix} ,
\]
up to a permutation of rows.
\end{enumerate}  
\end{thm}

We close this section by discussing the important special case where $w=0$ of Theorem \ref{thm:ExtrPosLkClinfty}. Note that each cone $\cC_{X}$ over a finite-dimensional normed space $X$ has a distinguished base cut out by the functional $\phi_X:\R\times X\ra \R$ given by $\phi_X(t,x)=t$. We call a map $L:\R\times X\ra \R\times Y$ \emph{dual-normalized} if $L\circ\phi_X=\phi_Y$. It is then easy to see that the set of dual-normalized maps in $\Pos\lb\cC_{\ell^n_\infty},\cL_k\rb$ coincides with the set $\Pos_0(\cL_k,\cC_{\ell^n_\infty})$ defined above. Theorem \ref{thm:ExtrPosLkClinfty} gives the following characterization of its extreme points.

\begin{lem}\label{lem:ExtremePointsTPLkCinfty}
A linear map $Q:\R^{k+1}\ra \R^{n+1}$ is extremal in the set of dual-normalized positive maps, i.e., in the set $\Pos_0(\cL_k,\cC_{\ell^n_\infty})$, if and only if there exists a partition $n=n_1+n_2$ for $n_1,n_2\in\N$ and a vector $s\in \lset \pm 1\rset^{n_1}$ and an extremal contraction $A:\ell^k_2\ra \ell^{n_2}_\infty$ such that, up to a permutation of rows, we have
\[
Q = \begin{pmatrix} 1 & 0 \\ s & 0 \\ 0 & A\end{pmatrix} = \begin{pmatrix} 1 & 0 \\ s & 0 \\ 0 & \one_{n_2}\end{pmatrix} \begin{pmatrix} 1 & 0 \\ 0 & A\end{pmatrix}.
\]
\end{lem}

\section{The set \texorpdfstring{$\maxEA_2(\cC_{\ell^2_\infty},\cL_3)$}{maxEA2(Cellinfty,L3)} is not convex}\label{app:maxEANotConv}

We have the following theorem characterizing the central elements of the max-entanglement annihilating maps $\maxEA_2(\cC_{\ell^2_\infty},\cL_k)$ for any $k\in\N$.

\begin{thm}\label{thm:appConvexityEA}
Let $\phi:\ell^2_\infty\ra \ell^n_2$ be a linear map and denote $\phi_i=\phi(e_i)\in \ell^n_2$ for the standard basis vectors $\lset e_1,e_2\rset\subseteq \ell^2_\infty$. For $\lambda\in\R$ we have $\lambda\oplus\phi\in\maxEA_2(\cC_{\ell^2_\infty},\cL_k)$ if and only if 
\[
\alpha(\phi) := \max\big\lbrace \| \phi_1 + \phi_2\|_{\ell^n_2},\| \phi_1 - \phi_2\|_{\ell^n_2}, \sqrt{\|\phi H \phi^* \|_1}, H\in\mathcal{H}  \big\rbrace \leq \lambda ,
\]
where the set $\mathcal{H}$ is given by
\[
    \mathcal{H}=\Big\lbrace \begin{pmatrix} -1 & 1 \\ 1 & 1\end{pmatrix},\begin{pmatrix} 1 & -1 \\ 1 & 1\end{pmatrix},\begin{pmatrix} 1 & 1 \\ -1 & 1\end{pmatrix},\begin{pmatrix} 1 & 1 \\ 1 & -1\end{pmatrix}\Big\rbrace .
\]
\end{thm}

\begin{proof}
The map $\lambda\oplus\phi$ is in $\Pos(\cC_{\ell^2_\infty},L_k)$ if and only if $\|\phi\|_{\ell^2_\infty\ra \ell^n_2}\leq \lambda$, which is equivalent to
\[
\max\Big\lbrace \| \phi_1 + \phi_2\|_{\ell^n_2},\| \phi_1 - \phi_2\|_{\ell^n_2}\Big\rbrace \leq 1.
\]
The extremal rays of the maximal tensor product $\cC_{\ell^2_\infty}\otimes_{\max} \cC_{\ell^2_\infty}$ either belong to the minimal tensor product $\cC_{\ell^2_\infty}\otimes_{\min} \cC_{\ell^2_\infty}$ or correspond to central maps $1\oplus H$ with $H:\ell^2_1\ra \ell^2_\infty$ such that $\pm H\in \mathcal{H}$. The map $\lambda\oplus\phi\in \Pos(\cC_{\ell^2_\infty},\cL_k)$ belongs to $\maxEA_2(\cC_{\ell^2_\infty},\cL_k)$ if and only if $\lambda^2 \oplus \phi H \phi^* \in \EB(\cL_k,\cL_k)$ for all $H\in\mathcal{H}$, or equivalently if $\|\phi H \phi^* \|_1 \leq \lambda^2$ for all $H\in\mathcal{H}$. Combining the condition for positivity and for being max-entanglement annihilating finishes the proof.
\end{proof}

Consider the linear maps $\phi^{(1)},\phi^{(2)}:\ell^2_\infty\ra \ell^2_2$ given by 
\[
\phi^{(1)} = \begin{pmatrix} 1 & 0 \\ 0 & 0.2 \end{pmatrix},
\]
and 
\[
\phi^{(2)} = \begin{pmatrix} 1 & 0.3 \\ 0 & 0.5 \end{pmatrix}.
\]
We find that
\[
\alpha(\phi^{(1)})+\alpha(\phi^{(2)})< \alpha(\phi^{(1)}+\phi^{(2)}) .
\]
By Theorem \ref{thm:appConvexityEA} we conclude that 
\[
\alpha(\phi^{(1)})\oplus \phi^{(1)} + \alpha(\phi^{(2)})\oplus \phi^{(2)} \notin \maxEA_2(\cC_{\ell^2_\infty},\cL_3) ,
\]
although $\alpha(\phi^{(1)})\oplus \phi^{(1)}\in \maxEA_2(\cC_{\ell^2_\infty},\cL_3)$ and $\alpha(\phi^{(2)})\oplus \phi^{(2)}\in \maxEA_2(\cC_{\ell^2_\infty},\cL_3)$. Hence, the set $\maxEA_2(\cC_{\ell^2_\infty},\cL_3)$ is not convex.

\section{The Lorentzian tensor product is not associative}\label{sec:NonAssoc}

Consider the tensors $z_\lambda\in \R^{n+1}\otimes \R^{n+1}\otimes \R^{n+1}$ given by 
\[
z_{\lambda} = \lambda e_0\otimes e_0\otimes e_0 + \sum^n_{i=1} e_i\otimes e_i\otimes e_i ,
\] 
with a parameter $\lambda\in \R_0^+$. 

\begin{prop}
The following statements hold:
\begin{enumerate}
\item We have $z_{\lambda}\in \cL_n\otimes_L \lb \cL_n\otimes_L \cC_{\ell^n_1}\rb$ if and only if $\lambda\geq 1$.
\item If $z_{\lambda}\in \lb \cL_n\otimes_L \cL_n\rb \otimes_L \cC_{\ell^n_1}$, then $\lambda\geq \sqrt{n}$.
\end{enumerate}
In particular, the Lorentzian tensor product is not associative.
\end{prop}

\begin{proof}
Since $\cL_n\otimes_L \cC = \cL_n\otimes_{\max} \cC$ for any proper cone $\cC$, we have
\[
\cL_n\otimes_L \lb \cL_n\otimes_L \cC_{\ell^n_1}\rb = \cL_n\otimes_{\max} \cL_n\otimes_{\max} \cC_{\ell^n_1}, 
\] 
by associativity of the maximal tensor product. We have
\[
(\ident_{n+1}\otimes \ident_{n+1}\otimes \alpha)(z_\lambda) = \alpha_0 \lambda e_0\otimes e_0 + \sum^n_{i=1} \alpha_i e_i\otimes e_i \in \cL_n\otimes_{\max} \cL_n ,
\]
for every $\alpha\in \cC_{\ell^n_\infty} = \cC^*_{\ell^n_1}$ if and only if $\lambda\geq 1$. This proves the first statement. 

To prove the second statement consider the linear symmetrization map $S:\R^{n+1}\otimes \R^{n+1}\ra \R^{n+1}\otimes \R^{n+1}$ given by 
\[
S(e_i\otimes e_j) =\begin{cases} e_i\otimes e_i, &\text{ if } i=j \\ 0, &\text{ otherwise,} \end{cases} 
\] 
and the linear map $D:\R^{n+1}\otimes \R^{n+1}\ra \R^{n+1}$ given by 
\[
D(e_i\otimes e_j) =\begin{cases} e_i, &\text{ if } i=j \\ 0, &\text{ otherwise.} \end{cases} 
\]
By~\cite[Lemma 6.2.]{aubrun2023annihilating}, we have $S(\cL_n\otimes_{\max} \cL_n) \simeq \cC_{\ell^n_2\otimes_{\epsilon} \ell^n_2}$. Since $\ell^n_2\otimes_{\epsilon} \ell^n_2\simeq S^{n}_\infty$, the space of $n\times n$ matrices equipped with the operator norm, we have $D(\cC_{\ell^n_2\otimes_{\epsilon} \ell^n_2})\subseteq \cC_{\ell^n_\infty}$. We conclude that $DS(\cL_n\otimes_{L} \cL_n) \subseteq \cC_{\ell^n_\infty}$. By definition of the Lorentzian tensor product we find that 
\[
(DS\otimes \ident_{n+1})\lb \lb \cL_n\otimes_L \cL_n\rb \otimes_L \cC_{\ell^n_1}\rb \subseteq \cC_{\ell^n_\infty} \otimes_L \cC_{\ell^n_1}.
\]
Now, assume that $z_\lambda\in \lb \cL_n\otimes_L \cL_n\rb \otimes_L \cC_{\ell^n_1}$ implying that 
\[
(DS\otimes \ident_{n+1})\lb z_\lambda \rb = \lambda e_0\otimes e_0 + \sum^n_{i=1} e_i\otimes e_i \in \cC_{\ell^n_\infty}\otimes_{L} \cC_{\ell^n_1}.
\]
We can identify this tensor with the central map $\lambda\oplus \ident_{\ell^n_1\ra \ell^n_1}\in \LorFact\lb \cC_{\ell^n_1},\cC_{\ell^n_1}\rb$. Using Theorem \ref{thm:Main1}, we find that $\lambda\geq \gamma_2(\ident_{\ell^n_1\ra \ell^n_1})=d(\ell^n_1,\ell^n_2) = \sqrt{n}$, where $d$ denotes the Banach-Mazur distance (see for example~\cite{tomczak1989banach}).  

\end{proof}

\bibliographystyle{alpha}
\bibliography{mybibliography.bib}

\end{document}